\documentclass[12pt]{article}

\usepackage[margin=1in]{geometry}
\usepackage{CJKutf8, amsmath, amssymb, amsthm, authblk, hyperref, enumitem, booktabs, orcidlink, longtable, bbm}
\usepackage{cleveref}

\newtheorem{conjecture}{Conjecture}
\newtheorem{lemma}{Lemma}
\newtheorem{theorem}{Theorem}

\theoremstyle{definition}
\newtheorem{definition}{Definition}

\theoremstyle{remark}
\newtheorem*{remark}{Remark}

\DeclareMathOperator{\Span}{span}
\DeclareMathOperator{\tr}{tr}
\DeclareMathOperator{\rank}{rank}
\DeclareMathOperator{\pf}{Pf}
\DeclareMathOperator*{\e}{\mathbb E}
\DeclareMathOperator{\poly}{poly}

\usepackage[style=trad-unsrt, citestyle=numeric-comp, giveninits=true, doi=false, url=false, eprint=false, isbn=false, backend=bibtex]{biblatex}

\begin{document}

\title{Quantum entropy thermalization}

\author[1,2]{Yichen Huang (黄溢辰)\orcidlink{0000-0002-8496-9251}\thanks{\href{mailto:huangtbcmh@gmail.com}{huangtbcmh@gmail.com}}}
\author[1]{Aram W. Harrow\orcidlink{0000-0003-3220-7682}\thanks{\href{mailto:aram@mit.edu}{aram@mit.edu}}}
\affil[1]{Center for Theoretical Physics, Massachusetts Institute of Technology, Cambridge, Massachusetts 02139, USA}
\affil[2]{Department of Physics, Harvard University, Cambridge, Massachusetts 02138, USA}

\begin{CJK}{UTF8}{gbsn}

\maketitle

\end{CJK}

\begin{abstract}

In an isolated quantum many-body system undergoing unitary evolution, the entropy of a subsystem (smaller than half the system size) thermalizes if at long times, it is to leading order equal to the thermodynamic entropy of the subsystem at the same energy. In this paper, we prove entropy thermalization for a nearly integrable Sachdev-Ye-Kitaev model initialized in a pure product state. The model is obtained by adding random all-to-all $4$-body interactions as a perturbation to a random free-fermion model.

\end{abstract}

Preprint number: MIT-CTP/5468

\tableofcontents

\section{Introduction}\label{sec:intro}

\subsection{Entropy thermalization} \label{ss:et}

The success of statistical mechanics is due in large part to the explanatory power of thermal equilibrium.  In particular, thermal equilibrium is not only simple to describe but also an accurate description of many systems in nature.  Its ubiquity comes from the fact that many systems that are far from thermal equilibrium tend to evolve towards thermal equilibrium. This process is called thermalization. Understanding the reasons for thermalization is a fundamental mystery in physics dating back to the origins of statistical mechanics and thermodynamics.

More specifically, we say that a system in contact with a bath thermalizes if it evolves to a Gibbs state described by the canonical ensemble. Suppose an isolated quantum many-body system is initialized in a pure state. Under unitary evolution the system stays in a pure state and never thermalizes. To observe thermalization, we divide the system into two parts. Let $A$ be a subsystem smaller than or equal to half the system size and $\bar A$ be the complement of $A$ (rest of the system). We view $\bar A$ as a bath of $A$ and consider the thermalization of $A$: At long times, do expectation values of observables on $A$ evolve to those in a Gibbs state?

In textbooks and introductory courses, we learned that the canonical ensemble can be derived from the ergodic hypothesis or the principle of equal a priori probabilities. The derivation is simple, elegant, and independent of the microscopic details of the system under consideration. However, the hypothesis or principle is an assumption that is usually difficult to verify or disprove for a particular system. Therefore, it would be desirable if thermalization (described in the preceding paragraph) could be derived from the fundamental laws of physics, in particular, the axioms of quantum mechanics.

In this paper, we study entropy thermalization.  We will see how this is an information-theoretic way of measuring thermalization.

\begin{definition} [Entropy thermalization] \label{def:etrpt}
The entropy of subsystem $A$ thermalizes if after long-time evolution, it is to leading order (in the size of $A$) equal to the thermodynamic entropy of $A$ at the same energy. The thermodynamic entropy of $A$ means the entropy of the reduced density matrix for $A$ of a thermal state of the whole system.
\end{definition}

Throughout this paper, entropy refers to the von Neumann entropy unless otherwise stated. 

Energy is conserved in the evolution under a (time-independent) Hamiltonian and is thus a macroscopic constraint. Although not exactly conserved, the energy of subsystem $A$ usually equilibrates in that its temporal fluctuations at long times vanish in the thermodynamic limit (equilibration can be proved under mild assumptions \cite{Rei08, LPSW09, Sho11}). Since the thermal state maximizes the entropy among all states with the same energy \cite{Weh78}, the entropy of $A$ is upper bounded by the thermodynamic entropy of $A$ at the same energy. Entropy thermalization (Definition \ref{def:etrpt}) says that equilibrated states saturate the upper bound to leading order and is thus a maximum entropy principle.

The above discussion can be rephrased in the language of free energy, which is minimized by the thermal state. Neglecting the temporal fluctuations of the energy of subsystem $A$, entropy thermalization means that the free energy of $A$ at long times is only slightly higher than the minimum such that their difference is subextensive in the size of $A$. In this sense, entropy thermalization is an information-theoretic statement that a subsystem looks thermal at long times.

Since chaotic dynamics tends to maximize the entropy, it is widely believed that:

\begin{conjecture}
Entropy thermalizes in quantum chaotic systems where energy is the only local conserved quantity.
\end{conjecture}

\begin{remark}
This conjecture is not mathematically precise because ``chaotic'' is not defined. We do not attempt to define it here, for there is no clear-cut definition of quantum chaos. There are also subtleties in defining subsystem energy because there are multiple ways to handle boundary terms in the Hamiltonian.  This will be discussed further at the end of \Cref{sec:intro}.
\end{remark}

Entanglement, a concept of quantum information theory, has been widely used in condensed matter and statistical physics to provide insights beyond those obtained via ``conventional'' quantities. A large body of literature is available on the static \cite{HLW94, VLRK03, RM04, KP06, LW06, Has07, HM14,  VHBR17} and dynamical \cite{CC05, ZPP08, BPM12, KH13} behavior of entanglement in various systems. Entanglement is of experimental interest \cite{IMP+15, KTL+16, LRS+19, BEJ+19}, and its scaling reflects the classical simulability of quantum many-body systems \cite{Vid03, VC06, SWVC08, Osb12, GHLS15}.

For a pure state, the entanglement entropy between subsystems $A$ and $\bar A$ is defined as the von Neumann entropy of $A$. Thus, Definition \ref{def:etrpt} is mathematically equivalent to

\begin{definition} [entanglement thermalization \cite{ZKH15}] \label{def:et}
The entanglement entropy between $A$ and $\bar A$ thermalizes if after long-time evolution, it is to leading order equal to the thermodynamic entropy of $A$ (the smaller subsystem) at the same energy.
\end{definition}

To study entanglement thermalization, we usually initialize the system in an unentangled (or slightly entangled) pure state.

Most studies on thermalization only consider local observables or subsystems of constant size. Since entropy is usually referred to as an extensive quantity, we must consider larger subsystems in order to study entropy scaling with subsystem size. Subsystem entropy as a function of subsystem size is called the ``Page curve,'' which was originally calculated for random states (corresponding to infinite temperature) \cite{Pag93}. In high-energy physics, the Page curve also refers to the entanglement between a black hole and its Hawking radiation \cite{Pag93I, AEMM19, Pen20, AHM+21}.  Our main results can be viewed as deriving Page curves for long-time evolved states (corresponding to any temperature) in an interacting quantum many-body system.

\subsection{Model} \label{ss:m}

Due to the difficulty of solving generic quantum many-body systems, significant analytical effort in condensed matter physics is devoted to constructing models in which desired physical properties can be shown rigorously. For example, the transverse-field Ising chain \cite{Pfe70} and the toric code \cite{Kit03} are representative models, which appear in tutorials on quantum phase transition and topological order, respectively. In this paper, we prove entropy thermalization in a nearly integrable Sachdev-Ye-Kitaev (SYK) model. To our knowledge, this is the first proof of entropy thermalization in a particular quantum system.

Over the past several years, SYK models \cite{SY93, Kit15, Sac15, MS16, GQS17, GKST20} have become an active research topic in condensed matter and high-energy physics. The real (complex) SYK$q$ model is a quantum mechanical model of Majorana (Dirac) fermions with random all-to-all $q$-body interactions (``$q$-body'' means that each term in the Hamiltonian acts non-trivially only on $q$ sites). For $q=2$, the model is an integrable model of free fermions \cite{Mag16}. For even $q\ge4$, the model is chaotic \cite{Kit15, MS16}. Entanglement dynamics in SYK models has been studied using a combination of analytical and numerical methods \cite{GLQ17, Mag17, CQZ20, Zha20, SNSF22, PSSY22, Zha22, KM17}.

Let $[N]:=\{1,2,\ldots,N\}$ be the set of integers from $1$ to $N$. Consider an $N$-mode fermionic system with creation and annihilation operators $a_j^\dag,a_j$ indexed by $j\in[N]$.

\begin{definition} [complex SYK2 model] \label{def:SYK2}
Let $h$ be a random matrix of order $N$ from the Gaussian unitary ensemble. The Hamiltonian of the complex SYK2 model is
\begin{equation} \label{eq:SYK2}
H_\textnormal{SYK2}=a^\dag ha,
\end{equation}
where $a:=(a_1,a_2,\ldots,a_N)^T$ is a column vector of $N$ operators.
\end{definition}

\begin{definition} [complex SYK4 model \cite{Sac15, GKST20}] \label{def:SYK4}
Let
\begin{equation} 
\mathcal I:=\left\{(j,k,l,m)\in[N]^{\times4}:(j<k)~\textnormal{and}~(l<m)~\textnormal{and}~(jN+k\le lN+m)\right\}
\end{equation}
and $J:=\{J_{jklm}\}_{(j,k,l,m)\in\mathcal I}$ be a collection of $|\mathcal I|$ independent complex Gaussian random variables with zero mean $\overline{J_{jklm}}=0$ and unit variance $\overline{|J_{jklm}|^2}=1$. The Hamiltonian of the complex SYK4 model is
\begin{equation}  \label{eq:SYK4}
H_\textnormal{SYK4}=\sum_{(j,k,l,m)\in\mathcal I}J_{jklm}a_j^\dag a_k^\dag a_la_m+\textnormal{H.c.},
\end{equation}
where ``H.c.'' means Hermitian conjugate.
\end{definition}

The complex SYK$q$ model is also known as the embedded Gaussian unitary ensemble \cite{Kot01, BW03} and has been studied under this name for decades.

Let
\begin{equation} \label{eq:Q}
Q:=\sum_{j=1}^Na_j^\dag a_j
\end{equation}
be the fermion number operator. Let $\epsilon_1,\epsilon_2$ be infinitesimal and
\begin{equation} \label{eq:SYK24}
H_\textnormal{SYK}:=H_\textnormal{SYK2}+\epsilon_2H_\textnormal{SYK4}.
\end{equation}
Our model is
\begin{equation} \label{eq:model}
H=Q+\epsilon_1H_\textnormal{SYK}=Q+\epsilon_1H_\textnormal{SYK2}+\epsilon_1\epsilon_2H_\textnormal{SYK4}.
\end{equation}
By definition, $H_\textnormal{SYK2}$ and $H_\textnormal{SYK4}$ and hence $H_\textnormal{SYK}$ and $H$ conserve fermion number in that
\begin{equation} \label{eq:comm}
[H_\textnormal{SYK2},Q]=[H_\textnormal{SYK4},Q]=[H_\textnormal{SYK},Q]=[H,Q]=0.
\end{equation}

Both $H_\textnormal{SYK}$ and $H$ are nearly integrable as both $H_\textnormal{SYK2}$ and $Q+\epsilon_1H_\textnormal{SYK2}$ are integrable.

\paragraph{Effects of infinitesimal perturbations.}The Hamiltonian (\ref{eq:model}) can be viewed as perturbing $Q$ with $\epsilon_1H_\textnormal{SYK}$, and $H_\textnormal{SYK}$ can be viewed as perturbing $H_\textnormal{SYK2}$ with $\epsilon_2H_\textnormal{SYK4}$. The effects of the infinitesimal perturbations on most but not all properties are infinitesimal.

Let
\begin{equation} \label{eq:thermal}
\sigma_\beta:=\frac{e^{-\beta H}}{\tr(e^{-\beta H})}
\end{equation}
be a thermal state at inverse temperature $\beta$. Since $\sigma_\beta$, as a matrix-valued function of $\epsilon_1,\epsilon_2$, is continuous,
\begin{equation}
\sigma_\beta=\frac{e^{-\beta Q}}{\tr(e^{-\beta Q})} 
\end{equation}
up to an infinitesimal error. Thus, the thermal properties of $H$ are infinitesimally close to those of $Q$.

Due to fermion number conservation (\ref{eq:comm}), $H$ and $H_\textnormal{SYK}$ have exactly the same set of eigenstates. Let $\{|j\rangle_h\}_{j=1}^{2^N}$ be a complete set of eigenstates of $a^\dag ha$. If the spectrum of $a^\dag ha$ is non-degenerate, first-order perturbation theory implies that the eigenbasis of $a^\dag ha+\epsilon_2H_\textnormal{SYK4}$ is $\{|j\rangle_h\}$ up to an infinitesimal error.

Although the integrability-breaking perturbation (SYK4 term in $H$) is infinitesimal, we will show that its effects on the dynamics become significant in the long-time limit $\lim_{\epsilon_1,\epsilon_2\to0}\lim_{t\to\infty}$. Because the perturbation typically has operator norm $O( \epsilon_1\epsilon_2)$ (neglecting the $N$-dependence), its effects on the dynamics can become significant only after $t \gtrsim 1/(\epsilon_1\epsilon_2)$ (see Lemma 1 in Ref.~\cite{HC15}).  However, this is only a lower bound, and it is possible that entropy thermalization would need significantly more time.


\subsection{Results (informal)} \label{ss:ri}

We initialize the system in a (random) product state, where each fermionic mode is either vacant or occupied. Let $n$ be the fermion number of the state and $\nu:=n/N$ be the filling fraction. Assume without loss of generality that $\nu\le1/2$, and suppose that $\nu$ is lower bounded by an arbitrarily small positive constant. Since the ensemble of SYK$q$ Hamiltonians is invariant with respect to permutations of indices, we may further assume without loss of generality that the initial state is
\begin{equation} \label{eq:ini}
|\phi\rangle=a_1^\dag a_2^\dag\cdots a_n^\dag|0\rangle,
\end{equation}
where $|0\rangle$ is the vacuum state with no fermions. Since $|\phi\rangle$ has a definite fermion number, $H$ and $H_\textnormal{SYK}$ generate the same dynamics in the sense that
\begin{equation} \label{eq:24}
e^{-iHt}|\phi\rangle=e^{-int}e^{-iH_\textnormal{SYK}\epsilon_1t}|\phi\rangle=e^{-int}e^{-i\epsilon_1H_\textnormal{SYK2}t-i\epsilon_1\epsilon_2H_\textnormal{SYK4}t}|\phi\rangle.
\end{equation}

\begin{theorem} [main result, informal]
In the model (\ref{eq:model}), entropy thermalizes in the sense of Definition \ref{def:etrpt}.
\end{theorem}

Specifically, we prove upper and lower bounds on the entanglement entropy at long times (Table \ref{t:summary}). The leading terms in both bounds are equal to the thermodynamic entropy of the subsystem at the same energy. Furthermore, the subleading terms in our upper and lower bounds match up to at most a logarithmic factor in $N$. Thus, both bounds are pretty tight and neither of them can be substantially improved.

\begin{table}
\caption{Summary of results in the thermodynamic limit $N\to\infty$. $L$ is the subsystem size. $\nu=n/N\le1/2$ and $f:=L/N\le1/2$ are fixed positive constants. $H_b(\nu):=-\nu\ln\nu-(1-\nu)\ln(1-\nu)$ is the binary entropy function. $C$'s in different cells are \emph{different} positive constants that depend only on $\nu$ and $f$. $\nu$ is related to temperature by (\ref{eq:temp}). The second and third columns are our upper (Theorem \ref{thm:ub}) and lower bounds (Theorems \ref{thm:half}, \ref{thm:l}) on the entanglement entropy at long times, respectively. The fourth column is the thermodynamic entropy (\ref{eq:the}) of the subsystem at the same energy. The last column is the average entanglement entropy of a random state with $n$ fermions (Eq.~(\ref{eq:random-ent}) \cite{BHK+22}).}
\centering
\begin{tabular}{c|cccc} 
\toprule
& upper bound & lower bound & thermal & random \\
\midrule
$\nu=1/2$ or $f\neq 1/2$ & $H_b(\nu)L-C$ & $H_b(\nu)L-C\ln L$  & $H_b(\nu)L$ & $H_b(\nu)L-C$ \\
$\nu\neq1/2$ and $f=1/2$ & $H_b(\nu)L-C\sqrt L$ & $H_b(\nu)L-C\sqrt{L\ln L}$ & $H_b(\nu)L$ & $H_b(\nu)L-C\sqrt L$\\
\bottomrule
\end{tabular}
\label{t:summary}
\end{table}

\subsection{Discussion}

\paragraph{Typicality.}As a quantum analogue of the principle of equal a priori probabilities, the typicality argument \cite{slloyd88, GLTZ06, PSW06} postulates that the properties of a quantum system are described by a typical state subject to macroscopic constraints (if any). It is widely believed that the argument leads to quantitatively accurate predictions in chaotic systems. In our context, the constraint is the energy or fermion number, which is conserved under time evolution. The argument claims that \emph{the entanglement entropy of $e^{-iHt}|\phi\rangle$ at long times is approximately equal to that of a random pure state with $n$ fermions.}

The average entanglement entropy of a random state with definite particle number was calculated exactly in Refs.~\cite{Hua19NPB, BD19, Hua21NPB, BHK+22}. For any $\nu$ and subsystem size, it falls within the narrow gap between our upper and lower bounds (Table \ref{t:summary}). Thus, the claim (in italics above) of the typicality argument is rigorously confirmed in a particular model.

Random quantum circuits are minimal models of quantum chaotic dynamics and are therefore of particular interest \cite{FKNV22}. In order to more faithfully model the energy-conserving dynamics generated by chaotic local Hamiltonians, a conserved quantity was introduced into random quantum circuits \cite{KVH18, RPv18}. It allows us to define entanglement thermalization for quantum circuits. We conjecture that entanglement thermalizes in random quantum circuits with a conserved quantity.

\paragraph{Quantum chaos and integrability breaking.}Quantum integrable systems are usually exactly or partially solvable. As platforms where physics can be demonstrated analytically, integrable systems are of particular interest, but they are very special. Weak integrability-breaking perturbations are always present in real quantum systems. They make an integrable system chaotic and can thus drastically change some physical properties. For example, integrable systems often support ballistic or dissipationless transport \cite{SS90, ZNP97, Pro11, SPA11, KBM12}, while transport in chaotic systems is generically diffusive with finite conductivity \cite{JHR06, JR07, HKM13, Zni20}. Therefore, it is important to understand the onset of quantum chaos due to integrability breaking \cite{SR10, LSPR21, BHG21}. What is the finite-size scaling of the critical perturbation strength above which physical properties become chaotic? The answer to this question may depend on the physical property under consideration.

Recall that $Q+\epsilon_1H_\textnormal{SYK2}$ is an integrable free-fermion model, in which entanglement does not thermalize \cite{BHK21, BHK+22} because $e^{-i(Q+\epsilon_1H_\textnormal{SYK2})t}|\phi\rangle$ for any $t\in\mathbb R$ is a fermionic Gaussian state \cite{Bra05}. By contrast, we have proved entanglement thermalization in the model (\ref{eq:model}). Thus, an infinitesimal integrability-breaking perturbation leads to chaotic entanglement dynamics at long times. This result is independent of the microscopic details of the perturbation as long as the perturbation is generic. For example, both upper and lower bounds in Table \ref{t:summary} hold for the entanglement entropy of $e^{-i(Q+\epsilon_1H_\textnormal{SYK2}+\epsilon_1\epsilon_2 H_\textnormal{d})t}|\phi\rangle$ at long times, where $H_\textnormal{d}$, defined in Eq.~(\ref{eq:dd}) below, is a model with random all-to-all $2$-body density-density interactions.

\paragraph{Thermalization versus eigenstate thermalization.}In a quantum many-body system, thermalization with respect to subsystem $A$ means that the reduced density matrix for $A$ at long times is approximately equal to that of a thermal state with the same energy (``approximately equal to'' means that the trace distance between the two reduced density matrices vanishes in the thermodynamic limit). The eigenstate thermalization hypothesis (ETH) \cite{Deu91, Sre94, RDO08, DLL18, Deu18} was proposed to explain thermalization. It says that the reduced density matrix for $A$ of a single eigenstate is approximately equal to that of a thermal state with the same energy. Under mild additional assumptions, the ETH implies thermalization \cite{GE16, DKPR16}.  Thermalization, in turn, implies entropy thermalization due to the continuity of entropy \cite{Fan73, Aud07}.

In a companion paper \cite{HH22TET}, we consider the model (\ref{eq:model}) with the initial state (\ref{eq:ini}). When the size of $A$ is larger than the square root of but is (up to a logarithmic factor) a vanishing fraction of the system size, we prove thermalization with high probability, while almost all eigenstates violate the ETH. In this sense, the ETH is not a necessary condition for thermalization. When the size of $A$ is a finite fraction of the system size, we disprove thermalization. Thus, we find a regime where entropy thermalization occurs without thermalization.

The ETH for entropy says that for individual eigenstates, the entropy of subsystem $A$ (smaller than or equal to half the system size) is to leading order equal to the thermodynamic entropy of $A$ at the same energy \cite{Deu10, SPR12, DLS13, GG18, HG19}. Is the ETH for entropy a necessary condition for entropy thermalization? No, and an explicit counterexample is the model (\ref{eq:model}). Indeed, the spectrum of $H_\textnormal{SYK2}$ is non-degenerate with probability $1$ (Lemma \ref{l:nd}). Then, due to fermion number conservation (\ref{eq:comm}) and since the perturbation $\epsilon_2H_\textnormal{SYK4}$ is infinitesimal, $H$, $H_\textnormal{SYK}$, and $H_\textnormal{SYK2}$ have the same set of eigenstates (up to an infinitesimal error). The eigenstates of $H_\textnormal{SYK2}$ are random Gaussian states with definite fermion number (Lemma \ref{l:haar}). When the size of $A$ is a finite fraction of the system size, they do not satisfy the ETH for entropy \cite{LCB18, ZLC20, LRV20, BHK+22}.

Table \ref{t:comp} summarizes the conditions on the subsystem size under which various phenomena occur or do not occur in the model (\ref{eq:model}). It is very similar to Table 1 of Ref.~\cite{HH22TET}.

\begin{table}
\caption{Conditions on the subsystem size $L$ under which various phenomena occur or do not occur in the model (\ref{eq:model}). We write $x\ll y$ if $x/y\to 0$ in the thermodynamic limit $N\to\infty$; $x\gtrsim y$ if $x/y$ is lower bounded by an arbitrarily small positive constant. While the ETH and the ETH for entropy are statements about the (static) Hamiltonian, thermalization and entropy thermalization are dynamic processes from an initial product state. \emph{Thermalization without the ETH} and \emph{entropy thermalization without the ETH for entropy} are new phenomena, which occur when $\sqrt N\lesssim L\ll N/\ln N$ and $N\lesssim L\le N/2$, respectively.}
\centering
\begin{tabular}{c|cccc}
\toprule
& ETH & thermalization & ETH for entropy & entropy thermalization \\
\midrule
occur & $L\ll\sqrt N$ \cite{HH22TET} & $L\ll N/\ln N$ \cite{HH22TET} & $L\ll N$ \cite{BHK+22} & $L\le N/2$ [this work] \\
not occur & $L\gtrsim\sqrt N$ \cite{HH22TET} & $L\gtrsim N$ \cite{HH22TET} & $L\gtrsim N$ \cite{BHK+22} & N/A \\
\bottomrule
\end{tabular}
\label{t:comp}
\end{table}

\paragraph{Subsystem Hamiltonian.}
A general Hamiltonian $\mathbf H$ of the whole system (union of subsystems $A$ and $\bar A$) can be split into three parts: $\mathbf H=\mathbf H_A+\mathbf H_{\bar A}+\mathbf H_\partial$, where $\mathbf H_{A(\bar A)}$ contains terms acting only on $A(\bar A)$, and $\mathbf H_{\partial}$ consists of boundary terms. Thermalization of $A$ means that the state of $A$ at long times is approximately equal to the reduced state $\sigma_{\beta,A}:=\tr_{\bar A}\sigma_\beta$ of the thermal state $\sigma_\beta:=e^{-\beta\mathbf H}/\tr(e^{-\beta\mathbf H})$ of the whole system. One might also consider the thermal state $\sigma'_{\beta,A}:=e^{-\beta\mathbf H_A}/\tr(e^{-\beta\mathbf H_A})$ of $\mathbf H_A$. Due to the presence of boundary terms, $\sigma_{\beta,A}$ and $\sigma'_{\beta,A}$ are usually far in trace distance \cite{GG18}. On the other hand, for quantum lattice systems with sufficiently fast decay of correlations, the reduced density matrices of $\sigma_{\beta,A}$ and $\sigma'_{\beta,A}$ for a region that is deep enough in the interior of $A$ are close; see Theorem 5 in Ref.~\cite{BK19}. For a Hamiltonian with all-to-all interactions like the SYK model, even extensive quantities such as the entropies of $\sigma_{\beta,A}$ and $\sigma'_{\beta,A}$ may have different scaling~\cite{HG19}. In our model (\ref{eq:model}), $\sigma_{\beta,A}$ and $\sigma'_{\beta,A}$ are the same up to an infinitesimal error because $\mathbf H_\partial$ is infinitesimal.

\section{Results (formal)} \label{s:rf}

Throughout this paper, standard asymptotic notation is used extensively. Let $g_1,g_2:\mathbb R^+\to\mathbb R^+$ be two functions. One writes $g_1(x)=O(g_2(x))$ if and only if there exist constants $M,x_0>0$ such that $g_1(x)\le Mg_2(x)$ for all $x>x_0$; $g_1(x)=\Omega(g_2(x))$ if and only if there exist constants $M,x_0>0$ such that $g_1(x)\ge Mg_2(x)$ for all $x>x_0$; $g_1(x)=\Theta(g_2(x))$ if and only if there exist constants $M_1,M_2,x_0>0$ such that $M_1g_2(x)\le g_1(x)\le M_2g_2(x)$ for all $x>x_0$; $g_1(x)=o(g_2(x))$ if and only if for any constant $M>0$ there exists a constant $x_0>0$ such that $g_1(x)<Mg_2(x)$ for all $x>x_0$. We use a tilde to hide a polylogarithmic factor: $\tilde O(g_1(x)):=O(g_1(x))\poly|\ln g_1(x)|$.

Let $A$ be an arbitrary subsystem of $L$ fermionic modes and $\bar A$ be the complement of $A$ (rest of the system) so that $A\otimes\bar A$ defines a bipartition of the system.

\begin{definition} [entanglement entropy] \label{def:ee}
The R\'enyi entanglement entropy $S_\alpha$ with index $\alpha\in(0,1)\cup(1,\infty)$ of a pure state $|\xi\rangle$ is defined as
\begin{equation} \label{eq:Sa}
S_\alpha(\xi_A):=\frac1{1-\alpha}\ln\tr(\xi_A^\alpha),
\end{equation}
where $\xi_A=\tr_{\bar A}|\xi\rangle\langle\xi|$ is the reduced density matrix. The von Neumann entanglement entropy is given by
\begin{equation} \label{eq:S}
S(\xi_A):=\lim_{\alpha\to1}S_\alpha(\xi_A)=-\tr(\xi_A\ln\xi_A).
\end{equation}
\end{definition}

Neglecting infinitesimal quantities, the reduced density matrix of subsystem $A$ is
\begin{equation} \label{eq:thr}
\sigma_{\beta,A}:=\tr_{\bar A}\sigma_\beta=e^{-\beta Q_A}/\tr(e^{-\beta Q_A}),
\end{equation}
where
\begin{equation} \label{eq:qa}
Q_A:=\sum_{j\in A}a_j^\dag a_j
\end{equation}
is the restriction of $Q$ to $A$. If $\sigma_\beta$ and $|\phi\rangle$ have the same energy, then
\begin{equation} \label{eq:temp}
\tr(\sigma_\beta H)=\langle\phi|H|\phi\rangle=n\implies\beta=\ln(1/\nu-1).
\end{equation}
The von Neumann entropy of $\sigma_{\beta,A}$ is
\begin{equation} \label{eq:the}
S(\sigma_{\beta,A})=-\tr(\sigma_{\beta,A}\ln\sigma_{\beta,A})=H_b(\nu)L,
\end{equation}
where $H_b$ is the binary entropy function defined in the caption of Table \ref{t:summary}. Thus, the thermodynamic entropy obeys a volume law with coefficient $H_b(\nu)$.

Let $L,m$ be positive integers such that $Lm$ is a multiple of $N$. Let $A_1,A_2,\ldots,A_m$ be $m$ possibly overlapping subsystems, each of which has exactly $L$ fermionic modes. Suppose that each fermionic mode in the system is in exactly $Lm/N$ out of these $m$ subsystems. Let
\begin{equation} \label{eq:phit}
\phi(t):=e^{-iHt}|\phi\rangle\langle\phi|e^{iHt},\quad\phi(t)_{A_j}:=\tr_{\bar A_j}\phi(t)
\end{equation}
be the state and its reduced density matrix at time $t$, respectively.

\begin{theorem} [upper bound] \label{thm:u}
For any $N,h,J,t$,
\begin{equation} \label{eq:t2}
\frac1m\sum_{j=1}^mS(\phi(t)_{A_j})\le H_b(\nu)L.
\end{equation}
\end{theorem}

Let
\begin{equation} \label{eq:ncn}
{[N]\choose n}:=\{R\subseteq[N]:|R|=n\}
\end{equation}
be the set of size-$n$ subsets of $[N]=\{1,2,\ldots,N\}$. Let $\{|\psi_R\rangle\}_{R\in{[N]\choose n}}$ be the complete set of (pairwise orthogonal) computational basis states with $n$ fermions, where
\begin{equation} \label{eq:psi}
|\psi_R\rangle:=\left(\prod_{j\in R}a_j^\dag\right)|0\rangle.
\end{equation}
Let
\begin{equation} \label{eq:psit}
\psi_R(t):=e^{-iHt}|\psi_R\rangle\langle\psi_R|e^{iHt},\quad\psi_R(t)_A:=\tr_{\bar A}\psi_R(t)
\end{equation}
be the state and its reduced density matrix at time $t$, respectively.

\begin{theorem} [upper bound] \label{thm:ub}
Let $\nu\le1/2$ and $f=L/N\le1/2$ be fixed positive constants. For any $h,J,t$,
\begin{multline} \label{eq:up}
\frac1{{N\choose n}}\sum_{R\in{[N]\choose n}}S(\psi_R(t)_A)\\
\le H_b(\nu)L-\delta_{f,\frac12}\left|\ln\frac{1-\nu}{\nu}\right|\sqrt{\frac{\nu(1-\nu)N}{2\pi}}+\frac{f+\ln(1-f)}2+\frac{O(\delta_{f,\frac12})}{\sqrt L}+\frac{O(1)}L
\end{multline}
in the limit $N\to\infty$, where $\delta$ is the Kronecker delta.
\end{theorem}

Since $H$ conserves fermion number, we observe that for any $t\in\mathbb R$, $\tr(\phi(t)Q)=n$ and $\{e^{-iHt}|\psi_R\rangle\}_{R\in{[N]\choose n}}$ is a complete orthonormal basis of the subspace
\begin{equation} \label{eq:Mn}
\mathcal M_n:=\Span\left\{|\psi_R\rangle:R\in{[N]\choose n}\right\}.
\end{equation}
Theorems \ref{thm:u} and \ref{thm:ub} follow from these observations. Thus, their proofs do not depend on the microscopic details of $H$ and apply universally to any (possibly time-dependent) Hamiltonian that conserves fermion number.

The quantum recurrence theorem \cite{BL57} says that under the dynamics generated by a time-independent Hamiltonian, a finite-size system will, after a sufficiently long time, return to a state arbitrarily close to the initial state, i.e., for any $H$ and any $\varepsilon,t_0>0$ there exists $t>t_0$ such that $|\langle\phi|e^{-iHt}|\phi\rangle|>1-\varepsilon$. Since $|\phi\rangle$ is a product state,
\begin{equation}
\inf_{t>t_0}S(\phi(t)_A)=0,\quad\forall t_0>0,\quad\phi(t)_A:=\tr_{\bar A}\phi(t).
\end{equation}
However, we prove non-trivial lower bounds on the entanglement entropy at most times.

Let $\tau$ be sufficiently large and $t$ be uniformly distributed in the interval $[0,\tau]$. Conceptually, $\tau$ needs to be sufficiently large such that the effect of the SYK4 term in Eq.~(\ref{eq:24}) is crucial at most time $t\in[0,\tau]$. At a technical level, the proofs of Theorems \ref{thm:half} and \ref{thm:l} below approximate the infinite-time average $\lim_{\tau'\to\infty}\e_{t\in[0,\tau']}$ by the long-time average $\e_{t\in[0,\tau]}$. $\tau$ needs to be sufficiently large such that the approximation error is negligible.

\begin{theorem} [lower bound at infinite temperature] \label{thm:half}
Suppose that $\nu=1/2$. Then,
\begin{multline} \label{eq:10}
\Pr_h\left(\Pr_J\left(\Pr_{t\in[0,\tau]}\left(\frac1m\sum_{j=1}^mS(\phi(t)_{A_j})=L\ln2-\frac{O(L\ln N)}N\right)=1-e^{-\Omega(N)}\right)=1\right)\\
\ge1-1/\poly(N)
\end{multline}
for $L\le N/10$, and
\begin{equation} \label{eq:50}
\Pr_h\left(\Pr_J\left(\Pr_{t\in[0,\tau]}\big(S_2(\phi(t)_A)=L\ln2-O(\ln N)\big)\ge1-\frac{e^{-\Omega(\frac N2-L)}}{\poly(N)}\right)=1\right)\ge1-\frac1{\poly(N)}
\end{equation}
for $L\le N/2$, where $\poly(N)$ denotes a polynomial of sufficiently high degree in $N$.
\end{theorem}
\begin{remark}
Since the R\'enyi entropy $S_\alpha$ is monotonically non-increasing in $\alpha$, (\ref{eq:50}) remains valid upon replacing $S_2(\phi(t)_A)$ by $S(\phi(t)_A)$.
\end{remark}

Table \ref{t:program} interprets (\ref{eq:10}) as a program. (\ref{eq:50}), (\ref{eq:less}), (\ref{eq:half}) can be interpreted similarly.

\begin{table}
\caption{Interpretation of (\ref{eq:10}) as a program.}
\centering
\begin{tabular}{l}
\toprule
step 1: Sample $h$, and proceed to the next step with probability $1-1/\poly(N)$.\\
step 2: Sample $J$, and proceed to the next step with probability $1$.\\
step 3: Sample $t\in[0,\tau]$, and proceed to the next step with probability $1-e^{-\Omega(N)}$.\\
step 4: Entanglement thermalizes in that $\sum_{j=1}^mS(\phi(t)_{A_j})/m=L\ln2-O(L\ln N)/N$.\\
\bottomrule
\end{tabular}
\label{t:program}
\end{table}

\begin{theorem} [lower bound at finite temperature] \label{thm:l}
Suppose that $1/2>\nu=\Omega(1)$. Then,
\begin{multline} \label{eq:less}
\Pr_h\left(\Pr_J\left(\Pr_{t\in[0,\tau]}\left(\frac1m\sum_{j=1}^mS(\phi(t)_{A_j})\ge H_b(\nu)L-\frac{O(L\ln N)}N\right)=1-e^{-\frac{\Omega(N/2-L)^2}N}\right)=1\right)\\
\ge1-1/\poly(N)
\end{multline}
for $L\le N/2-c\sqrt{N\ln N}$, where $c$ is a sufficiently large constant, and
\begin{multline} \label{eq:half}
\Pr_h\bigg(\Pr_J\left(\Pr_{t\in[0,\tau]}\left(S(\phi(t)_A)\ge H_b(\nu)L-O(\sqrt{L\ln N})\right)\ge1-\min\left\{e^{-\frac{\Omega(N/2-L)^2}N},\frac1{\poly(N)}\right\}\right)\\
=1\bigg)\ge1-\frac1{\poly(N)}
\end{multline}
for $L\le N/2$.
\end{theorem}

Theorems \ref{thm:half} and \ref{thm:l} are the main technical contributions of this paper.

Since the Hamiltonian $H$ conserves fermion number, $\phi(t)$ for any $t$ is a pure state with $n$ fermions. Thus, it is instructive to compare our bounds with the subsystem  entropy of a typical pure state and of the maximally mixed state, both with definite fermion number. (In the case of the pure state, this means entanglement entropy.)  Let $\xi_A:=\tr_{\bar A}|\xi\rangle\langle\xi|$ be the reduced density matrix of a state chosen uniformly at random with respect to the Haar measure from $\mathcal M_n$. Exact formulas for $\e S(\xi_A)$ and $S(\e\xi_A)$ are given by Eq.~(54) of Ref.~\cite{BHK+22} and Eq.~(13) of Ref.~\cite{VR17}, respectively. 

\begin{theorem} \label{thm:r}
  Let $\nu\le1/2$ and $f\le1/2$ be fixed positive constants, and consider the limit $N\to\infty$.
  \begin{enumerate}
  \item
    The expected entanglement entropy of a random pure state with definite particle number is
\begin{equation}
  \e_{|\xi\rangle\in\mathcal M_n}S(\xi_A)=H_b(\nu)L-\delta_{f,\frac12}\left|\ln\frac{1-\nu}{\nu}\right|\sqrt{\frac{\nu(1-\nu)N}{2\pi}}+\frac{f+\ln(1-f)-\delta_{\nu,\frac12}\delta_{f,\frac12}}2+o(1).
  \label{eq:random-ent}
\end{equation}
\item
  The subsystem entropy of the maximally mixed state with definite particle number is
\begin{equation}
  S\left(\e_{|\xi\rangle\in\mathcal M_n}\xi_A\right)
  = H_b(\nu) L 
  +\frac{f+\ln(1-f)}2+o(1).
  \label{eq:subsystem-mixed}
\end{equation}
\end{enumerate}
\end{theorem}

The right-hand side of Eq.~(\ref{eq:random-ent}) is very similar to that of (\ref{eq:up}).


\section{Proof sketches}

In this section, we sketch the proofs of our results. Full details are deferred to \Cref{sec:UB-proofs,app:B}.

\subsection{Upper bounds}

\paragraph{Proof sketch of Theorem \ref{thm:u}.}Since $\tr(\phi(t)Q)=n$ for any $t$ and since the thermal state maximizes the von Neumann entropy among all states with the same energy \cite{Weh78}, we obtain (\ref{eq:t2}) by viewing $Q$ as an artificial Hamiltonian. We average over subsystems $A_1,A_2,\ldots,A_m$ in (\ref{eq:t2}) because the average subsystem fermion number is $\nu L$ (invariant under time evolution).

\paragraph{Proof sketch of Theorem \ref{thm:ub}.}$\{e^{-iHt}|\psi_R\rangle\}_{R\in{[N]\choose n}}$ for any $t$ is a complete orthonormal basis of $\mathcal M_n$. Unlike (\ref{eq:t2}), (\ref{eq:up}) holds for every single subsystem $A$. We do not need to average over subsystems in (\ref{eq:up}) because the expectation value of $Q_A$, averaged over $\{\psi_R(t)\}_{R\in{[N]\choose n}}$, is $\nu L$ (invariant under time evolution). The proof of Theorem \ref{thm:u} (sketched above) implies that
\begin{equation} \label{eq:upw}
\frac1{{N\choose n}}\sum_{R\in{[N]\choose n}}S(\psi_R(t)_A)\le H_b(\nu)L.
\end{equation}

Comparing the right-hand sides of (\ref{eq:up}) and (\ref{eq:upw}), we see that they are the same to leading order in $L$. The former contains negative subleading terms and thus improves the latter. We now explain why the improvement is possible. For any $t$, the average $\sum_{R\in{[N]\choose n}}\psi_R(t)/{N\choose n}$ is the maximally mixed state on $\mathcal M_n$. Measuring $Q_A$ in this state results in a hypergeometric distribution with mean $\nu L$: The probability of observing $j$ fermions in $A$ is ${L\choose j}{N-L\choose n-j}/{N\choose n}$. By contrast, measuring $Q_A$ in a thermal state results in a binomial distribution. The discrepancy between hypergeometric and binomial distributions implies that (\ref{eq:upw}) cannot be saturated.

The calculations in the proof of Theorem \ref{thm:ub} have substantial overlap with but are significantly different from those in previous works \cite{VR17, Hua19NPB}. This will be discussed further in \Cref{app:ub}.

\subsection{Lower bounds} \label{ss:l}

The proofs of (\ref{eq:10}), (\ref{eq:50}), (\ref{eq:less}), (\ref{eq:half}) are similar at a high level but differ in some technical details. To provide an overview, here we only sketch the proof of (\ref{eq:less}). Complete proofs of all results are given in \Cref{app:B}.

\begin{lemma} \label{l:nd}
The spectrum of $H_\textnormal{SYK2}$ is almost surely (with respect to the randomness of $h$) non-degenerate.
\end{lemma}

This lemma allows us to assume that the spectrum of $a^\dag ha$ is non-degenerate. Then, the eigenbasis $\{|j\rangle_h\}_{j=1}^{2^N}$ of $a^\dag ha$ is unambiguously defined. For infinitesimal $\epsilon_2$, $H_\textnormal{SYK}=a^\dag ha+\epsilon_2H_\textnormal{SYK4}$ also has a non-degenerate spectrum, and its eigenbasis is $\{|j\rangle_h\}$ up to an infinitesimal error, which is neglected in the following. Let
\begin{equation} \label{eq:bphi}
\bar\phi:=\lim_{\tau\to\infty}\frac1\tau\int_0^\tau\phi(t)\,\mathrm dt=\lim_{\tau\to\infty}\frac1\tau\int_0^\tau e^{-iH_\textnormal{SYK}t}|\phi\rangle\langle\phi|e^{iH_\textnormal{SYK}t}\,\mathrm dt,\quad\bar\phi_{A_j}:=\tr_{\bar A_j}\bar\phi
\end{equation}
be the time-averaged state and its reduced density matrix, respectively. Since the spectrum of $H_\textnormal{SYK}$ is non-degenerate, $\bar\phi$ is obtained by dephasing $|\phi\rangle$ in the eigenbasis of $H_\textnormal{SYK}$ or in the basis $\{|j\rangle_h\}$. Since $\langle\phi|j\rangle_h=0$ if the fermion number of $|j\rangle_h$ is not $n$,
\begin{equation} \label{eq:dep}
\bar\phi=\sum_{j=1}^{2^N}\big|\langle\phi|j\rangle_h\big|^2|j\rangle_h\langle j|_h=\sum_{j:~Q|j\rangle_h=n|j\rangle_h}\big|\langle\phi|j\rangle_h\big|^2|j\rangle_h\langle j|_h.
\end{equation}
The effective dimension of $|\phi\rangle$ is defined as
\begin{equation} \label{eq:eff}
D^\textnormal{eff}_h=e^{S_2(\bar\phi)}=1\Big/\sum_{j=1}^{2^N}\big|\langle\phi|j\rangle_h\big|^4\implies D^\textnormal{eff}_h\le e^{S(\bar\phi)}\le{N\choose n},\quad\forall h.
\end{equation}
Since $h$ is chosen from the Gaussian unitary ensemble, $|j\rangle_h$ is a uniformly random Gaussian state with definite fermion number (Lemma \ref{l:haar}). Using random matrix techniques, we show that with high probability (Lemma \ref{l:eff}),
\begin{equation} \label{eq:sk}
D^\textnormal{eff}_h\ge{N\choose n}\big/\poly(N).
\end{equation}
Thus, $S(\bar\phi)$ obeys a volume law with coefficient $H_b(\nu)$. Using the strong subadditivity of the von Neumann entropy \cite{LR73, AM03}, $\sum_{j=1}^mS(\bar\phi_{A_j})/m$ is to leading order lower bounded by $H_b(\nu)L$. Lemma \ref{l:ndg} allows us to assume that the spectrum of $H_\textnormal{SYK}$ has non-degenerate gaps (Definition \ref{d:ndg}). Then, (\ref{eq:sk}) implies equilibration in that $\phi(t)_{A_j}$ is close to $\bar\phi_{A_j}$ for most $t\in[0,\tau]$ (Lemma \ref{l:eq}). Finally, we obtain the lower bound (\ref{eq:less}) on $\sum_{j=1}^mS(\phi(t)_{A_j})/m$ from the continuity of the von Neumann entropy (Lemma \ref{l:cont}).

\begin{remark}
  As aforementioned, infinitesimal $\epsilon_2$ implies that the eigenbasis of $a^\dag ha+\epsilon_2H_\textnormal{SYK4}$ is infinitesimally close to $\{|j\rangle_h\}$. This is essential for the proof. If $\epsilon_2$ is not extraordinarily small, of course entropy thermalization is still expected. However, we do not have an analytical expression for the eigenstates of $H_\textnormal{SYK2}+\epsilon_2H_\textnormal{SYK4}$ and do not know how to estimate the effective dimension, which should be defined with respect to the eigenbasis of $H_\textnormal{SYK2}+\epsilon_2H_\textnormal{SYK4}$.
\end{remark}

\section{Proof of upper bounds} \label{sec:UB-proofs}

For spin systems, Theorems \ref{thm:dpn} and \ref{thm:dpnb} below provide upper bounds on the entanglement entropy: Theorem \ref{thm:dpn} applies to individual states with a certain expectation value of the particle number, and Theorem \ref{thm:dpnb} applies to a complete set of orthonormal basis states with definite particle number. Without essential modification the proofs and bounds remain valid for fermionic systems. Since $H$ conserves fermion number, for any $t\in\mathbb R$, $\tr(\phi(t)Q)=n$ and $\{e^{-iHt}|\psi_R\rangle\}_{R\in{[N]\choose n}}$ is a complete orthonormal basis of $\mathcal M_n$. Thus, Theorems \ref{thm:u} and \ref{thm:ub} follow directly from the fermionic analogues of Theorems \ref{thm:dpn} and \ref{thm:dpnb}, respectively.

Consider a system of $N$ spin-$1/2$'s (qubits) labeled by $1,2,\ldots,N$. Let $\nu:=n/N$.

\subsection{Entropy of a single state}

Let $L,m$ be positive integers such that $Lm$ is a multiple of $N$. Let $A_1,A_2,\ldots,A_m$ be $m$ possibly overlapping subsystems, each of which has exactly $L$ spins. Suppose that each spin in the system is in exactly $Lm/N$ out of these $m$ subsystems. For each $j$, let $\bar A_j$ be the complement of $A_j$ so that $A_j\otimes\bar A_j$ defines a bipartition of the system.

\begin{theorem} \label{thm:dpn}
Let
\begin{equation}
Z:=\sum_{k=1}^NZ_k,
\end{equation}
where $Z_k$ is the Pauli $z$ matrix acting on spin $k$. For any state $|\xi\rangle$ such that $\langle\xi|Z|\xi\rangle=2n-N$,
\begin{equation}
\frac1m\sum_{j=1}^mS(\xi_{A_j})\le H_b(\nu)L,
\end{equation}
where $\xi_{A_j}:=\tr_{\bar A_j}|\xi\rangle\langle\xi|$ is the reduced density matrix of subsystem $A_j$.
\end{theorem}

\begin{proof}
Regarding $Z$ as a Hamiltonian, let
\begin{equation}
Z_{A_j}:=\sum_{k\in A_j}Z_k,\quad E_{A_j}=\tr(\xi_{A_j}Z_{A_j})
\end{equation}
be the restriction of $Z$ to subsystem $A_j$ and the energy of $A_j$, respectively. Since each spin is in exactly $Lm/N$ out of the $m$ subsystems $A_1,A_2,\ldots,A_m$,
\begin{equation}
\frac1m\sum_{j=1}^mZ_{A_j}=LZ/N\implies\frac1m\sum_{j=1}^mE_{A_j}=(2\nu-1)L.
\end{equation}
Since the thermal state maximizes the von Neumann entropy among all states with the same energy \cite{Weh78},
\begin{equation}
\frac1m\sum_{j=1}^mS(\xi_{A_j})\le\frac1m\sum_{j=1}^mH_b\left(\frac12+\frac{E_{A_j}}{2L}\right)L\le H_b\left(\frac12+\frac1m\sum_{j=1}^m\frac{E_{A_j}}{2L}\right)L=H_b(\nu)L,
\end{equation}
where we used the concavity of $H_b$.
\end{proof}

\subsection{Average entropy over a basis} \label{app:ub}

Let $M_n$ be the set of computational basis states with $n$ spins up and $N-n$ spins down. Let $M'_n$ be an arbitrary orthonormal basis of $\Span M_n$ so that
\begin{equation}
|M_n|=|M'_n|={N\choose n}.
\end{equation}
Let $A$ be an arbitrary subsystem of $L$ spins and $\bar A$ be the complement of $A$.

\begin{theorem} \label{thm:dpnb}
Let $\nu=n/N\le1/2$ and $f:=L/N\le1/2$ be fixed positive constants. In the limit $N\to\infty$,
\begin{equation} \label{eq:dpnb}
\e_{|\xi\rangle\in M'_n}S(\xi_A)\le H_b(\nu)L-\delta_{f,\frac12}\left|\ln\frac{1-\nu}{\nu}\right|\sqrt{\frac{\nu(1-\nu)N}{2\pi}}+\frac{f+\ln(1-f)}2+\frac{O(\delta_{f,\frac12})}{\sqrt L}+\frac{O(1)}L,
\end{equation}
where $\xi_A:=\tr_{\bar A}|\xi\rangle\langle\xi|$ is the reduced density matrix of subsystem $A$.
\end{theorem}

\begin{proof}
Let $M_j^A$ ($M_j^{\bar A}$) be the set of computational basis states of subsystem $A$ ($\bar A$) with $j$ spins up and $L-j$ ($N-L-j$) spins down so that
\begin{equation}
|M_j^A|={L\choose j},\quad|M_j^{\bar A}|={N-L\choose j},\quad M_n=\bigcup_{j=0}^{\min\{L,n\}}M_j^A\times M_{n-j}^{\bar A}.
\end{equation}
Thus, any state $|\xi\rangle$ in $M'_n$ can be decomposed as
\begin{equation}
|\xi\rangle=\sum_{j=0}^{\min\{L,n\}}c_j|\xi_j\rangle,
\end{equation}
where $|\xi_j\rangle$ is a normalized state in $\Span M_j^A\otimes\Span M_{n-j}^{\bar A}$. Let $\xi_{j,A}:=\tr_{\bar A}|\xi_j\rangle\langle\xi_j|$ be the reduced density matrix of subsystem $A$. Since $\Span M_n^{\bar A},\Span M_{n-1}^{\bar A},\ldots$ are pairwise orthogonal,
\begin{gather}
  \xi_A=\bigoplus_{j=0}^{\min\{L,n\}}|c_j|^2\xi_{j,A},\\
   S(\xi_A)=\sum_{j=0}^{\min\{L,n\}}|c_j|^2S(\xi_{j,A})-\sum_{j=0}^{\min\{L,n\}}|c_j|^2\ln(|c_j|^2).
\end{gather}
Since $M'_n$ is a complete basis of $\Span M_n$,
\begin{equation}
\e_{|\xi\rangle\in M'_n}|c_j|^2=p_j,\quad p_j:=|M_j^A||M_{n-j}^{\bar A}|/|M_n|.
\end{equation}
Using the trivial bound
\begin{equation}
S(\xi_{j,A})\le\ln\min\{|M_j^A|,|M_{n-j}^{\bar A}|\}
\end{equation}
and the concavity of the Shannon entropy,
\begin{equation}
\e_{|\xi\rangle\in M'_n}S(\xi_A)\le\sum_{j=0}^{\min\{L,n\}}p_j\ln\min\{|M_j^A|,|M_{n-j}^{\bar A}|\}-\sum_{j=0}^{\min\{L,n\}}p_j\ln p_j.
\end{equation}

For any fixed constant $f<1/2$,
\begin{equation} \label{eq:37}
\e_{|\xi\rangle\in M'_n}S(\xi_A)\le\sum_{j=0}^{\min\{L,n\}}p_j\ln\frac{|M_n|}{|M_{n-j}^{\bar A}|}\le\sum_{j:~|j/L-\nu|=\tilde O(1/\sqrt L)}p_j\ln\frac{|M_n|}{|M_{n-j}^{\bar A}|}+\frac1{\poly(L)},
\end{equation}
where we used the tail bound for the hypergeometric distribution \cite{Ska13}. Stirling's formula
\begin{equation}
\ln(N!)=N\ln N-N+\frac12\ln(2\pi N)+O(1/N)
\end{equation}
implies that
\begin{equation} \label{eq:bn}
{N\choose n}=\frac{N!}{n!(N-n)!}=e^{H_b(\nu)N+O(1/n)}\sqrt{\frac N{2\pi n(N-n)}}.
\end{equation}
Hence,
\begin{equation} \label{eq:40}
\ln\frac{|M_n|}{|M_{n-j}^{\bar A}|}=H_b(\nu)N-H_b\left(\frac{n-j}{N-L}\right)(N-L)+\frac12\ln\frac{(n-j)(N-L-n+j)}{(1-\nu)n(N-L)}+\frac{O(1)}{n-j}.
\end{equation}
It is not difficult to see that
\begin{gather} 
\sum_{j:~|j/L-\nu|=\tilde O(1/\sqrt L)}p_j\left(H_b(\nu)N+\frac{O(1)}{n-j}\right)=H_b(\nu)N+O(1/L),\label{eq:41}\\
\sum_{j:~|j/L-\nu|=\tilde O(1/\sqrt L)}\frac{p_j}2\ln\frac{(n-j)(N-L-n+j)}{(1-\nu)n(N-L)}=\frac{\ln(1-f)}2+O(1/L).\label{eq:42}
\end{gather}
Using the Taylor expansion
\begin{equation} \label{eq:tl}
H_b(\nu+x)=H_b(\nu)+x\ln\frac{1-\nu}{\nu}-\frac{x^2}{2\nu(1-\nu)}+\frac{(1-2\nu)x^3}{6\nu^2(1-\nu)^2}+O(x^4),
\end{equation}
a straightforward order-by-order calculation shows that
\begin{multline} \label{eq:exp}
(N-L)\sum_{j:~|j/L-\nu|=\tilde O(1/\sqrt L)}p_jH_b\left(\frac{n-j}{N-L}\right)\\
=\underbrace{H_b(\nu)(N-L)\vphantom{\frac1(}}_\textnormal{contribution from the zeroth-}+\underbrace{0\vphantom{\frac1(}}_\textnormal{first-}-\underbrace{\frac{L}{2(N-1)}}_\textnormal{second-order term in the Taylor expansion (\ref{eq:tl})}+\underbrace{O(1/L)\vphantom{\frac1(}}_\textnormal{total error},
\end{multline}
where we used the exact formula for the variance of the hypergeometric distribution \cite{Ska13}
\begin{equation}
\sum_{j=0}^{\min\{L,n\}}p_j(j-\nu L)^2=\frac{nL(N-L)(N-n)}{N^2(N-1)}
\end{equation}
in calculating the contribution from the second-order term in Eq.~(\ref{eq:tl}). We complete the proof for $f<1/2$ by combining (\ref{eq:37}), (\ref{eq:40}), (\ref{eq:41}), (\ref{eq:42}), (\ref{eq:exp}).

For $f=1/2$, assume without loss of generality that $n$ is odd (the case of even $n$ can be handled similarly) so that 
\begin{equation}
\e_{|\xi\rangle\in M'_n}S(\xi_A)\le2\sum_{j=0}^{(n-1)/2}p_j\ln\frac{|M_n|}{|M_{n-j}^{\bar A}|}\le2\sum_{j:~0<\nu-2j/N=\tilde O(1/\sqrt N)}p_j\ln\frac{|M_n|}{|M_{n-j}^{\bar A}|}+\frac1{\poly(N)}.
\end{equation}
We perform a similar calculation as in the case $f<1/2$ with the following modifications. Equation (\ref{eq:42}) becomes
\begin{equation}
2\sum_{j:~0<\nu-2j/N=\tilde O(1/\sqrt N)}\frac{p_j}2\ln\frac{(n-j)(N/2-n+j)}{(1-\nu)nN/2}=-\frac{\ln2}2+O(1/\sqrt L).
\end{equation}
Moreover, the contribution from the first-order term in Eq.~(\ref{eq:tl}) to Eq.~(\ref{eq:exp}) is non-zero. The contribution is
\begin{multline}
\frac N2\left(\ln\frac{1-\nu}{\nu}\right)2\sum_{j:~0<\nu-2j/N=\tilde O(1/\sqrt N)}p_j\left(\frac{n-j}{N/2}-\nu\right)\\
=2\left(\ln\frac{1-\nu}{\nu}\right){N\choose n}^{-1}\sum_{k=0.5,1.5,2.5,\ldots,\tilde O(\sqrt N)}k{N/2\choose n/2-k}{N/2\choose n/2+k},
\end{multline}
where we define $k=n/2-j$. Using Eqs.~(\ref{eq:bn}), (\ref{eq:tl}),
\begin{align}
&{N\choose n}^{-1}\sum_{k=0.5,1.5,2.5,\ldots,\tilde O(\sqrt N)}k{N/2\choose n/2-k}{N/2\choose n/2+k}\nonumber\\
&=\sum_{k=0.5,1.5,2.5,\ldots,\tilde O(\sqrt N)}k\sqrt{\frac{2Nn(N-n)}{\pi(n^2-4k^2)((N-n)^2-4k^2)}}e^{\frac N2\left(H_b\left(\nu-\frac{2k}N\right)+H_b\left(\nu+\frac{2k}N\right)-2H_b(\nu)\right)+O(1/N)}\nonumber\\
&=\sqrt{\frac{2 N}{\pi n(N-n)}}\sum_{k=0.5,1.5,2.5,\ldots,\tilde O(\sqrt N)}ke^{-\frac{2k^2}{\nu(1-\nu)N}}\left(1+O(k^4/N^3+k^2/N^2+1/N)\right)\nonumber\\
&=\sqrt{\frac{2 N}{\pi\nu(1-\nu)}}\left(\sum_{k=0.5,1.5,2.5,\ldots,\infty}\frac k{\sqrt N}e^{-\frac2{\nu(1-\nu)}\left(\frac k{\sqrt N}\right)^2}\frac1{\sqrt N}\right)+O(1/\sqrt N).
\end{align}
The last sum in the big parentheses is a Riemann sum for the integral
\begin{equation}
\int_0^\infty xe^{-\frac{2x^2}{\nu(1-\nu)}}\,\mathrm dx=\frac{\nu(1-\nu)}4.
\end{equation}
The approximation error of the Riemann sum is $O(1/N)$. We complete the proof for $f=1/2$ by combining the calculations above.
\end{proof}

\begin{remark}
A bound very similar to (\ref{eq:dpnb}) was derived in Ref.~\cite{VR17}, and the calculations in the proof of Theorem \ref{thm:dpnb} have substantial overlap with those in Refs.~\cite{VR17, Hua19NPB}. The differences are

\begin{itemize} [nosep]
\item The bound of Ref.~\cite{VR17} is for random states with definite particle number and normally distributed real coefficients in the computational basis, while our bound (\ref{eq:dpnb}) is for an arbitrary set of complex orthonormal basis states with definite particle number.
\item Both Refs.~\cite{VR17, Hua19NPB} approximate the hypergeometric distribution by a Gaussian distribution. However, neither reference attempts to upper bound the approximation error, which is necessary for ensuring the correctness of higher-order terms. By contrast, the proof of Theorem \ref{thm:dpnb} does not use the Gaussian approximation and is thus free of such error.
\end{itemize}
\end{remark}

Equation (\ref{eq:subsystem-mixed}) was derived in Ref.~\cite{VR17}, but the derivation again uses the Gaussian approximation without bounding the error. We give a proof of Eq.~(\ref{eq:subsystem-mixed}) by modifying that of Theorem \ref{thm:dpnb}.

\begin{proof} [Proof of Eq.~(\ref{eq:subsystem-mixed})]
Using the notation of the proof of Theorem \ref{thm:dpnb}, for any $0<f\le1/2$,
\begin{equation}
S\left(\e_{|\xi\rangle\in\mathcal M_n}\xi_A\right)=\sum_{j=0}^{\min\{L,n\}}p_j\ln\frac{|M_n|}{|M_{n-j}^{\bar A}|}.
\end{equation}
Then, Eq.~(\ref{eq:subsystem-mixed}) follows from this equation and Eqs.~(\ref{eq:40}), (\ref{eq:41}), (\ref{eq:42}), (\ref{eq:exp}).
\end{proof}

\section{Proof of lower bounds} \label{app:B}

In this section, we use the notation of \Cref{ss:l}.

\subsection{Spectrum}

\begin{definition} [non-degenerate spectrum] 
The spectrum of a Hamiltonian is non-degenerate if all eigenvalues are distinct.
\end{definition}

\begin{proof} [Proof of Lemma \ref{l:nd}]
The spectrum of $\sum_{j=1}^Nh_ja_j^\dag a_j$ is non-degenerate if $h_1,h_2,\ldots,h_N$ are linearly independent real numbers over the integers. Then, Lemma \ref{l:nd} can be proved in the same way as Lemma 2 in Ref.~\cite{KLW15}.
\end{proof}

\begin{definition} [non-degenerate gap] \label{d:ndg}
The spectrum $\{E_j\}$ of a Hamiltonian has non-degenerate gaps if the differences $\{E_j-E_k\}_{j\neq k}$ are all distinct, i.e., for any $j\neq k$,
\begin{equation} \label{eq:ndg}
E_j-E_k=E_{j'}-E_{k'}\implies(j=j')~\textnormal{and}~(k=k').
\end{equation}
\end{definition}

By definition, the non-degenerate gap condition implies that the spectrum is non-degenerate.

Let $K:=\{K_{jk}\}_{1\le j\le k\le N}$ be a collection of $N(N+1)/2$ continuous real random variables. Define
\begin{equation} \label{eq:dd}
H_\textnormal{d}:=\sum_{1\le j\le k\le N}K_{jk}a_j^\dag a_ja_k^\dag a_k.
\end{equation}
Lemma \ref{l:ndgd} below holds for almost any distribution of $K$.  For concreteness, let the elements of $K$ be independent real Gaussian random variables with zero mean $\overline{K_{jk}}=0$ and unit variance $\overline{K_{jk}^2}=1$.

\begin{lemma} \label{l:ndgd}
For any fixed Hermitian matrix $h'$ of order $N$ and any $\epsilon_2\in\mathbb R\setminus\{0\}$, the spectrum of $H_\textnormal{SYK2d}:=a^\dag h'a+\epsilon_2H_\textnormal{d}$ almost surely (with respect to the randomness of $K$) has non-degenerate gaps.
\end{lemma}

\begin{proof}
We begin by following the proof of Lemma 8 in Ref.~\cite{HH19}. Let $\{E_j\}_{j=1}^{2^N}$ be the eigenvalues of $H_\textnormal{SYK2d}$ and
\begin{equation}
G:=\prod_{((j\neq j')\textnormal{ or }(k\neq k'))\textnormal{ and }((j\neq k')\textnormal{ or }(k\neq j'))}(E_j+E_k-E_{j'}-E_{k'})
\end{equation}
so that $G=0$ if and only if the spectrum of $H_\textnormal{SYK2d}$ has degenerate gaps. It is easy to see that $G$ is a symmetric polynomial in $E_1,E_2,\ldots,E_{2^N}$. The fundamental theorem of symmetric polynomials implies that $G$ can be expressed as a polynomial in $F_1,F_2,\ldots$, where
\begin{equation}
F_k:=\sum_{j=1}^{2^N}E_j^k=\tr(H_\textnormal{SYK2d}^k).
\end{equation}
We see that every $F_k$ and hence $G$ are polynomials in the elements of $K$. Since the zeros of a multivariate polynomial are of measure zero unless the polynomial is identically zero, it suffices to find a particular sample of $K$ such that the spectrum of $H_\textnormal{SYK2d}$ has non-degenerate gaps. Let $K$ be a collection of real numbers that are linearly independent over the integers. Lemma \ref{l:ndgc} below implies that the spectrum of $H_\textnormal{d}$ has non-degenerate gaps. Hence, for sufficiently large $\Lambda$, the spectrum of $a^\dag h'a+\Lambda H_\textnormal{d}$ has non-degenerate gaps.
\end{proof}

\begin{lemma} \label{l:ndgc}
Let $p,q,r,s\in\{0,1\}^{\times N}$ be binary row vectors of length $N$ and $0_N$ denote the zero matrix of size $N\times N$. Then,
\begin{equation}
p^Tp-q^Tq=r^Tr-s^Ts\neq0_N\implies(p=r)~\textnormal{and}~(q=s).
\end{equation}
\end{lemma}

\begin{proof}
Let $p_1,q_1,r_1,s_1$ be the first elements of $p,q,r,s$, respectively. Since $p\neq q$, they differ by at least one element. Assume without loss of generality that $p_1=1$ and $q_1=0$. Then,
\begin{equation}
r_1-s_1=1\implies(r_1=1)~\textnormal{and}~(s_1=0)\implies p=r\implies q=s. 
\end{equation}
\end{proof}

Similar to Lemma \ref{l:ndgd}, we have

\begin{lemma} \label{l:ndg}
For any fixed Hermitian matrix $h'$ of order $N$ and any $\epsilon_2\in\mathbb R\setminus\{0\}$, the spectrum of $a^\dag h'a+\epsilon_2H_\textnormal{SYK4}$ almost surely (with respect to the randomness of $J$) has non-degenerate gaps.
\end{lemma}

\subsection{Gaussian states}

Define Majorana operators
\begin{equation} \label{eq:l}
    \chi_{2j-1}:=a_j+a_j^\dag,\quad\chi_{2j}:=i(a_j-a_j^\dag)
\end{equation}
with anticommutation relations $\{\chi_j,\chi_k\}=2\delta_{jk}$. Let
\begin{equation} \label{eq:etachi}
\eta:=(a_1,a_2,\ldots,a_N,a_1^\dag,a_2^\dag,\ldots,a_N^\dag)^T,\quad\chi:=(\chi_1,\chi_2,\ldots,\chi_{2N})^T
\end{equation}
be column vectors of $2N$ operators. The linear transformation (\ref{eq:l}) can be written as $\chi=Y\eta$, where $Y/\sqrt2$ is a unitary matrix of order $2N$.

For a density operator $\rho$, let $\langle B\rangle:=\tr(\rho B)$ denote the expectation value of an operator $B$. Let $\mathbf M$ be the covariance matrix in the Majorana representation with its elements given by
\begin{equation} \label{eq:covm}
    \mathbf M_{jk}=i\langle[\chi_j,\chi_k]\rangle/2.
\end{equation}
It is easy to see that $\mathbf M$ is a real antisymmetric matrix of order $2N$. Let $I_{N'}$ be the identity matrix of order $N'\in\mathbb N$. The covariance matrix can be written as
\begin{equation} \label{eq:72}
\mathbf M=i\langle\chi\chi^\dag\rangle-iI_{2N}=iY(\langle\eta\eta^\dag\rangle-I_{2N}/2)Y^\dag.
\end{equation}

If $\rho$ has a definite fermion number, then $\langle a_ja_k\rangle=\langle a_j^\dag a_k^\dag\rangle=0$ for any $j,k\in[N]$. Thus,
\begin{equation}
\langle\eta\eta^\dag\rangle-I_{2N}/2=\begin{pmatrix}\mathbf C&0\\0&-\mathbf C^*\end{pmatrix}
\end{equation}
is block diagonal, where (recall that $a=(a_1,a_2,\ldots,a_N)^T$ is a column vector of $N$ operators)
\begin{equation} \label{eq:def}
\mathbf C=\langle aa^\dag\rangle-I_N/2.
\end{equation}
is a Hermitian matrix of order $N$. The fermion number is
\begin{equation} \label{eq:pnum}
\sum_{j=1}^N\langle a_j^\dag a_j\rangle=N/2-\tr\mathbf C.
\end{equation}

\begin{definition} [fermionic Gaussian state]
A fermionic state is Gaussian if its density operator can be written as
\begin{equation}
\rho=e^{i\chi^TW\chi/2}/\tr(e^{i\chi^TW\chi/2})
\end{equation}
for some real antisymmetric matrix $W$ of order $2N$. Note that the elements of $W$ are allowed to be infinite.
\end{definition}

Since the covariance matrix $\mathbf M$ is antisymmetric, its eigenvalues $\{\pm i\lambda_1,\pm i\lambda_2,\ldots,\pm i\lambda_N\}$ are purely imaginary and come in conjugate pairs. A Gaussian state $\rho$ is pure if and only if $\lambda_j=\pm1$ for all $j$ \cite{Bra05}. If $\rho$ has a definite fermion number, $\rho$ is pure if and only if all eigenvalues of $\mathbf C$ are $\pm1/2$. Then, Eq.~(\ref{eq:pnum}) implies that the fermion number of $\rho$ is the number of negative eigenvalues of $\mathbf C$.

The overlap between two Gaussian states $\rho_1,\rho_2$ with covariance matrices $\mathbf M_1,\mathbf M_2$ is \cite{BG17}
\begin{equation}
\tr(\rho_1\rho_2)=|\pf(\mathbf M_1+\mathbf M_2)|/2^N,
\end{equation}
where $\pf$ denotes the Pfaffian of a matrix. If both $\rho_1$ and $\rho_2$ have definite fermion numbers, define $\mathbf C_1$ and $\mathbf C_2$ as in Eq.~(\ref{eq:def}) for $\rho_1$ and $\rho_2$, respectively. Since $\mathbf M_1+\mathbf M_2$ is antisymmetric and $Y/\sqrt2$ is unitary,
\begin{equation} \label{eq:overlap}
\tr(\rho_1\rho_2)=\frac{\sqrt{\det(\mathbf M_1+\mathbf M_2)}}{2^N}=\frac1{2^N}\sqrt{(2i)^{2N}\det\begin{pmatrix}\mathbf C_1+\mathbf C_2&0\\0&-\mathbf C_1^*-\mathbf C_2^*\end{pmatrix}}=|\det(\mathbf C_1+\mathbf C_2)|.
\end{equation}

\begin{definition} [uniformly random pure Gaussian state with definite fermion number] \label{def:haar}
A pure Gaussian state $\rho$ with fermion number $n$ is uniformly random if its correlation matrix (\ref{eq:def}) is given by
\begin{equation} \label{eq:haar}
\mathbf C=\frac12U^\dag\begin{pmatrix}-I_n&0\\0&I_{N-n}\end{pmatrix}U,
\end{equation}
where $U$ is a random unitary matrix uniformly distributed with respect to the Haar measure.
\end{definition}

It is not difficult to see that the eigenstates of $H_\textnormal{SYK2}$ are Gaussian states with definite fermion number. Furthermore,

\begin{lemma} [\cite{BHK+22}] \label{l:haar}
An eigenstate of $H_\textnormal{SYK2}$ with fermion number $n$ is a uniformly random Gaussian state in the sense of Definition \ref{def:haar}.
\end{lemma}

\begin{proof}
Suppose that $|\xi_1\rangle$ is an eigenstate of $a^\dag ha$. For an arbitrary unitary matrix $V$ of order $N$, let $|\xi_2\rangle$ be the state obtained from $|\xi_1\rangle$ with the substitution $a\to Va$ so that $|\xi_2\rangle$ is an eigenstate of $a^\dag V^\dag hVa$ with the same eigenvalue. Let $\mathbf C_1$ and $\mathbf C_2$ be the correlation matrices (\ref{eq:def}) of $|\xi_1\rangle$ and $|\xi_2\rangle$, respectively. Then,
\begin{equation}
\mathbf C_1=V\mathbf C_2V^\dag.
\end{equation}
Since $h$ is chosen from the Gaussian unitary ensemble, $h$ and $V^\dag hV$ have the same probability density. Therefore, $\mathbf C_1$ and $V^\dag \mathbf C_1V$ have the same probability density for any $V$.
\end{proof}

\subsection{Effective dimension}

Let
\begin{equation} \label{eq:ldef}
\mu_\alpha=\e_h\sum_{j=1}^{2^N}\big|\langle\phi|j\rangle_h\big|^{2\alpha}=\e_h\sum_{j:~Q|j\rangle_h=n|j\rangle_h}\big|\langle\phi|j\rangle_h\big|^{2\alpha}=\e_he^{(1-\alpha)S_\alpha(\bar\phi)},\quad\alpha>0.
\end{equation}
Normalization implies that $\mu_1=1$.

\begin{lemma}
For $n\le N/2$,
\begin{equation} \label{eq:la}
\mu_\alpha={N\choose n}\prod_{j=1}^n\frac{\Gamma(j+\alpha)\Gamma(N-n+j)}{\Gamma(j)\Gamma(N-n+j+\alpha)},
\end{equation}
where $\Gamma$ is the gamma function. In particular,
\begin{equation} \label{eq:l2}
\mu_2={N\choose n}^{-1}\frac{(n+1)(N-n+1)}{N+1}.
\end{equation}
\end{lemma}

\begin{proof}
It is easy to see that $|\phi\rangle$ is a Gaussian state with correlation matrix
\begin{equation}
\mathbf C_0=\frac12\begin{pmatrix}-I_n&0\\0&I_{N-n}\end{pmatrix}.
\end{equation}
Using Lemma \ref{l:haar} and Eq.~(\ref{eq:overlap}),
\begin{equation} \label{eq:lalpha}
\mu_\alpha={N\choose n}\e_U|\det(\mathbf C_0+U^\dag\mathbf C_0U)|^\alpha={N\choose n}\e_U|\det(I_N-P-U^\dag PU)|^\alpha,
\end{equation}
where $P:=I_N/2-\mathbf C_0$ is a Hermitian projector of rank $n$. Since $U$ is a Haar-random unitary, Theorem 1.1 in Ref.~\cite{Kar12} states that $P+U^\dag PU$ almost surely has $N-2n$ eigenvalues $0$. The other $2n$ eigenvalues are given by $\{1\pm\sqrt{x_j}\}_{j=1}^n$ with the joint probability density function
\begin{equation}
\mathbbm p(x_1,x_2,\ldots,x_n)\propto{\prod_{1\le j<k\le n}(x_j-x_k)^2}\prod_{j=1}^n(1-x_j)^{N-2n},\quad0\le x_j\le1.
\end{equation}
Hence, the eigenvalues of $I_N-P-U^\dag PU$ are $\pm\sqrt x_1,\ldots,\pm\sqrt x_n,1,\ldots,1$ so that
\begin{equation} \label{eq:88}
\e_U|\det(I_N-P-U^\dag PU)|^\alpha=\e_{(x_1,x_2,\ldots,x_n)\sim\mathbbm p}\prod_{j=1}^nx_j^\alpha.
\end{equation}
Since $\mathbbm p(x_1,x_2,\ldots,x_n)$ is the probability density function of the Jacobi unitary ensemble, Eq.~(2.15) of Ref.~\cite{Rou07} states that
\begin{equation} \label{eq:this}
\e_{(x_1,x_2,\ldots,x_n)\sim\mathbbm p}\prod_{j=1}^nx_j^\alpha=\prod_{j=1}^n\frac{\Gamma(j+\alpha)\Gamma(N-n+j)}{\Gamma(j)\Gamma(N-n+j+\alpha)},
\end{equation}
which can be easily derived from the Selberg integral \cite{FW08}. Equation (\ref{eq:la}) follows from Eqs.~(\ref{eq:lalpha}), (\ref{eq:88}), (\ref{eq:this}).
\end{proof}

\begin{lemma} \label{l:eff}
For any constant $\alpha>0$,
\begin{equation} \label{eq:ent}
\Pr_h\big(S_\alpha(\bar\phi)\ge H_b(\nu)N-O(\ln N)\big)\ge1-1/\poly(N).
\end{equation}
In particular,
\begin{equation}
\Pr_h\big(D^\textnormal{eff}_h\ge e^{H_b(\nu)N}/\poly(N)\big)\ge1-1/\poly(N).
\end{equation}
\end{lemma}

\begin{proof}
Assume without loss of generality that $n\le N/2$. Using the monotonicity of the R\'enyi entropy, it suffices to prove (\ref{eq:ent}) for any integer $\alpha\ge2$. For such $\alpha$, Eq.~(\ref{eq:la}) implies that
\begin{equation} 
\mu_\alpha=\prod_{j=1}^{\alpha-1}\frac{{n+j\choose j}}{{N+j\choose n}}\le e^{-(\alpha-1)H_b(\nu)N}\poly(n),
\end{equation}
where we used Eq.~(\ref{eq:bn}). Finally, (\ref{eq:ent}) follows from Eq.~(\ref{eq:ldef}) and Markov's inequality.
\end{proof}

\subsection{Equilibration}

The time-averaged expectation value of a (not necessarily Hermitian) linear operator $B$ is
\begin{equation} \label{eq:bb}
\bar B:=\lim_{\tau\to\infty}\frac1\tau\int_0^\tau\tr\big(\phi(t)B\big)\,\mathrm dt=\tr(\bar\phi B),
\end{equation}
and the fluctuation is
\begin{equation} \label{eq:fl}
\Delta B:=\lim_{\tau\to\infty}\frac1\tau\int_0^\tau\big|\tr\big(\phi(t)B\big)-\bar B\big|^2\,\mathrm dt.
\end{equation}

\begin{lemma} [\cite{Tas98, Rei08, LPSW09, Sho11}] \label{l:op}
If the spectrum of $H$ has non-degenerate gaps, then
\begin{equation}
\Delta B\le\|B\|^2/D^\textnormal{eff}_h.
\end{equation}
\end{lemma}

\begin{proof}
We include the proof of this lemma for completeness. Let $B_{jk}:=\langle j|_hB|k\rangle_h$ be the matrix element in the energy eigenbasis and $c_j:=\langle\phi|j\rangle_h$. Let $E_j$ be the eigenvalue corresponding to the eigenvector $|j\rangle_h$. Writing out the matrix elements,
\begin{align}
&\Delta B=\lim_{\tau\to\infty}\frac1\tau\int_0^\tau\big|\langle\phi|e^{iHt}Be^{-iHt}|\phi\rangle-\bar B\big|^2\,\mathrm dt=\lim_{\tau\to\infty}\frac1\tau\int_0^\tau\left|\sum_{j\neq k}c_jc_k^*B_{jk}e^{i(E_j-E_k)t}\right|^2\,\mathrm dt\nonumber\\
&=\sum_{j\neq k,~j'\neq k'}c_jc_k^*c_{j'}^*c_{k'}B_{jk}(B^\dag)_{k'j'}\lim_{\tau\to\infty}\frac1\tau\int_0^\tau e^{i(E_j-E_k-E_{j'}+E_{k'})t}\,\mathrm dt=\sum_{j\neq k}|c_j|^2|c_k|^2B_{jk}(B^\dag)_{kj}\nonumber\\
&\le\sqrt{\sum_{j\neq k}|c_j|^4B_{jk}(B^\dag)_{kj}\times\sum_{j\neq k}|c_k|^4(B^\dag)_{kj}B_{jk}}\le\sqrt{\sum_j|c_j|^4(BB^\dag)_{jj}\times\sum_k|c_k|^4(B^\dag B)_{kk}}\nonumber\\
&\le\|B\|^2\sum_j|c_j|^4=\|B\|^2/D^\textnormal{eff}_h,
\end{align}
where we used the non-degenerate gap condition and the fact that $(BB^\dag)_{jj}\le\|B\|^2$.
\end{proof}

Let
\begin{equation}
\|B\|_1:=\tr\sqrt{B^\dag B},\quad\|B\|_2:=\sqrt{\tr(B^\dag B)}
\end{equation}
denote the trace and Frobenius norms, respectively. It is well known that
\begin{equation} \label{eq:12norm}
\|B\|_2\le\|B\|_1\le\sqrt{\rank B}\|B\|_2.
\end{equation}

\begin{lemma} \label{l:eq}
Let $A'_1,A'_2,\ldots,A'_k$ be $k$ arbitrary subsystems, each of which has $L$ fermionic modes. For $\nu\le1/2$ and $L\le N/2-c\sqrt{N\ln N}$ with a sufficiently large constant $c$,
\begin{multline}
\Pr_h\left(\Pr_J\left(\lim_{\tau\to\infty}\frac1\tau\int_0^\tau\frac1k\sum_{j=1}^k\|\phi(t)_{A'_j}-\bar\phi_{A'_j}\|_1\,\mathrm dt\le\frac{d'}{\sqrt{D^\textnormal{eff}_h}}+e^{-\frac{\Omega(N/2-L)^2}L}\right)=1\right)\\
\ge1-1/\poly(N),
\end{multline}
where ($c'>0$ is an arbitrarily small constant)
\begin{equation} \label{eq:d}
d'=O(e^{H_b(\nu)L+c'(N/2-L)\ln(1/\nu-1)}).
\end{equation}
\end{lemma}

\begin{proof}
Let $A$ be an arbitrary subsystem of $L$ fermionic modes. Let $i_1<i_2<\cdots<i_L$ be the indices of the $L$ fermionic modes in $A$. Define
\begin{equation} \label{eq:PAdef}
P_A^{>l}=\sum_{(n_1,n_2,\ldots,n_L)\in\{0,1\}^{\times L}:~\sum_{j=1}^Ln_j>l}\prod_{j=1}^L\big(n_j+(1-2n_j)a_{i_j}a_{i_j}^\dag\big),\quad P_A^{\leqslant l}=\mathbbm 1_A-P_A^{>l},
\end{equation}
where $\mathbbm 1_A$ is the identity operator on $A$. Recall that $\{|\psi_R\rangle\}_{R\in{[N]\choose n}}$ is the complete set of computational basis states with $n$ fermions. By construction, $P_A^{>l}=(P_A^{>l})^2$ is a projector such that
\begin{equation}
P_A^{>l}|\psi_R\rangle=|\psi_R\rangle~\textnormal{or}~0
\end{equation}
if $|\psi_R\rangle$ has or does not have more than $l$ fermions in $A$, respectively. Let
\begin{equation}
l=\nu L+c'(N/2-L).
\end{equation}
Let $\e_{|A|=L}$ denote averaging over all subsystems of $L$ fermionic modes. There are ${N\choose L}$ such subsystems. Using the tail bound for the hypergeometric distribution \cite{Ska13},
\begin{equation} \label{eq:after-avg-def}
\e_{|A|=L}\|P_A^{>l}|\psi_R\rangle\|={N\choose L}^{-1}\sum_{j'>l}{n\choose j'}{N-n\choose L-j'}=e^{-\frac{\Omega(N/2-L)^2}L}.
\end{equation}
Any state $|\xi\rangle$ with $n$ fermions can be expanded as $|\xi\rangle=\sum_{R\in{[N]\choose n}}c'_R|\psi_R\rangle$ so that
\begin{equation} \label{eq:trunc1}
\e_{|A|=L}\|P_A^{>l}|\xi\rangle\|^2=\e_{|A|=L}\sum_{R\in{[N]\choose n}}|c'_R|^2\|P_A^{>l}|\psi_R\rangle\|=\sum_{R\in{[N]\choose n}}|c'_R|^2\e_{|A|=L}\|P_A^{>l}|\psi_R\rangle\|=e^{-\frac{\Omega(N/2-L)^2}L}.
\end{equation}
Using Eqs.~(\ref{eq:dep}), (\ref{eq:ldef}), (\ref{eq:l2}), (\ref{eq:trunc1}),
\begin{align}
&\e_h\tr(\bar\phi_AP_A^{>l})=\e_h\tr(\bar\phi P_A^{>l})=\e_h\sum_{j:~Q|j\rangle_h=n|j\rangle_h}\big|\langle\phi|j\rangle_h\big|^2\|P_A^{>l}|j\rangle_h\|^2\nonumber\\
&\le\sqrt{\mu_2\e_h\sum_{j:~Q|j\rangle_h=n|j\rangle_h}\|P_A^{>l}|j\rangle_h\|^4}=\sqrt{\mu_2\e_{|A|=L}\e_h\sum_{j:~Q|j\rangle_h=n|j\rangle_h}\|P_A^{>l}|j\rangle_h\|^4}\nonumber\\
&\le\sqrt{n+1}e^{-\frac{\Omega(N/2-L)^2}L}
\end{align}
so that
\begin{equation}
\e_h\frac1k\sum_{j=1}^k\tr(\bar\phi_{A'_j}P_{A'_j}^{>l})=\sqrt{n+1}e^{-\frac{\Omega(N/2-L)^2}L}.
\end{equation}
Markov's inequality implies that for $L\le N/2-c\sqrt{N\ln N}$ with a sufficiently large constant $c$,
\begin{equation} \label{eq:100}
\Pr_h\left(\frac1k\sum_{j=1}^k\tr(\bar\phi_{A'_j}P_{A'_j}^{>l})=e^{-\frac{\Omega(N/2-L)^2}L}\right)\ge1-1/\poly(N).
\end{equation}
Combining this inequality with Lemma \ref{l:ndg}, it suffices to prove that
\begin{equation} \label{eq:toprove}
\lim_{\tau\to\infty}\frac1\tau\int_0^\tau\frac1k\sum_{j=1}^k\|\phi(t)_{A'_j}-\bar\phi_{A'_j}\|_1\,\mathrm dt\le\frac{d'}{\sqrt{D^\textnormal{eff}_h}}+e^{-\frac{\Omega(N/2-L)^2}L}
\end{equation}
under the assumptions that (1) the spectrum of $H$ has non-degenerate gaps and (2)
\begin{equation} \label{eq:trunc2}
\frac1k\sum_{j=1}^k\tr(\bar\phi_{A'_j}P_{A'_j}^{>l})=e^{-\frac{\Omega(N/2-L)^2}L}.
\end{equation}
Lemma \ref{l:op} implies that
\begin{equation}
\lim_{\tau\to\infty}\frac1\tau\int_0^\tau\left|\tr\big(\phi(t)_AP_A^{>l}\big)-\tr(\bar\phi_AP_A^{>l})\right|\,\mathrm dt\le\frac1{\sqrt{D^\textnormal{eff}_h}}
\end{equation}
so that
\begin{equation} \label{eq:103}
\lim_{\tau\to\infty}\frac1\tau\int_0^\tau\frac1k\sum_{j=1}^k\tr\big(\phi(t)_{A'_j}P_{A'_j}^{>l}\big)\,\mathrm dt\le\frac1{\sqrt{D^\textnormal{eff}_h}}+e^{-\frac{\Omega(N/2-L)^2}L}.
\end{equation}
Although not a state with definite fermion number, $\bar\phi_A$ is a mixture of states, each of which has a definite fermion number. Hence,
\begin{equation} \label{eq:104}
\bar\phi_A=P_A^{>l}\bar\phi_AP_A^{>l}+P_A^{\leqslant l}\bar\phi_AP_A^{\leqslant l}\implies\tr(\bar\phi_AP_A^{>l})=\|\bar\phi_A-P_A^{\leqslant l}\bar\phi_AP_A^{\leqslant l}\|_1.
\end{equation}
Similarly,
\begin{equation} \label{eq:105}
\tr\big(\phi(t)_AP_A^{>l}\big)=\|\phi(t)_A-P_A^{\leqslant l}\phi(t)_AP_A^{\leqslant l}\|_1,\quad\forall t\in\mathbb R.
\end{equation}

A substantial part of the rest of the proof follows the calculation in Section 5 of Ref.~\cite{Sho11}. Let $\vec n=(n_1,n_2,\ldots,n_L)$ and $\vec n'=(n'_1,n'_2,\ldots,n'_L)$ be binary vectors of length $L$. Define
\begin{equation}
B_{\vec n}^{\vec n'}=\prod_{j=1}^L(a_{i_j}^\dag)^{n'_j}\times\prod_{j=1}^La_{i_j}a_{i_j}^\dag\times\prod_{j=1}^La_{i_j}^{n_j}.
\end{equation}
Any linear operator $B$ on subsystem $A$ can be expanded in the basis $\{B_{\vec n}^{\vec n'}\}_{\vec n,\vec n'\in\{0,1\}^{\times L}}$. Furthermore,
\begin{equation}
P_A^{\leqslant l}BP_A^{\leqslant l}\in\mathcal B,\quad\mathcal B:=\Span\left\{B_{\vec n}^{\vec n'}:{\sum_{j=1}^Ln'_j\le l\textnormal{ and }\sum_{j=1}^Ln_j\le l}\right\}.
\end{equation}
Let $g:\{0,1,\ldots,d-1\}\to\mathcal G$ be an arbitrary one-to-one mapping, where
\begin{equation}
\mathcal G:=\left\{\vec n\in\{0,1\}^{\times L}:\sum_{j=1}^Ln_j\le l\right\},\quad d:=|\mathcal G|=\sum_{j=0}^l{L\choose j}.
\end{equation}
Define $d^2$ operators
\begin{equation}
B_{dj_1+j_2+1}=\frac1{\sqrt d}\sum_{j=0}^{d-1}e^{2\pi ij_2j/d}B_{g(j)}^{g((j_1+j)\bmod d)}
\end{equation}
for $j_1,j_2=0,1,\ldots,d-1$. It is easy to see that
\begin{equation} \label{eq:orth}
\|B_j\|=1/\sqrt d,\quad\tr(B_j^\dag B_{j'})=\delta_{jj'},\quad\mathcal B=\Span\{B_1,B_2,\ldots,B_{d^2}\}.
\end{equation}
Since $P_A^{\leqslant l}(\phi(t)_A-\bar\phi_A)P_A^{\leqslant l}\in\mathcal B$, it can be expanded as
\begin{equation}
P_A^{\leqslant l}\big(\phi(t)_A-\bar\phi_A\big)P_A^{\leqslant l}=\sum_{j=1}^{d^2}c_j(t)B_j,\quad c_j(t)=\tr\left(B_j^\dag\big(\phi(t)_A-\bar\phi_A\big)\right).
\end{equation}
Lemma \ref{l:op} implies that
\begin{equation} \label{eq:113}
\lim_{\tau\to\infty}\frac1\tau\int_0^\tau|c_j(t)|^2\,\mathrm dt\le\frac{\|B_j^\dag\|^2}{D^\textnormal{eff}_h}=\frac1{dD^\textnormal{eff}_h}.
\end{equation}
Using (\ref{eq:12norm}), (\ref{eq:orth}),
\begin{equation}
\left\|P_A^{\leqslant l}\big(\phi(t)_A-\bar\phi_A\big)P_A^{\leqslant l}\right\|_1\le\sqrt{\rank P_A^{\leqslant l}\times\sum_{j,j'=1}^{d^2}c_j^*(t)c_{j'}(t)\tr(B_j^\dag B_{j'})}=\sqrt{d\sum_{j=1}^{d^2}|c_j(t)|^2}
\end{equation}
so that
\begin{equation} \label{eq:115}
\lim_{\tau\to\infty}\frac1\tau\int_0^\tau\left\|P_A^{\leqslant l}\big(\phi(t)_A-\bar\phi_A\big)P_A^{\leqslant l}\right\|_1\,\mathrm dt\le\sqrt{\sum_{j=1}^{d^2}\lim_{\tau\to\infty}\frac d\tau\int_0^\tau|c_j(t)|^2\,\mathrm dt}\le\frac d{\sqrt{D^\textnormal{eff}_h}}.
\end{equation}
Finally, (\ref{eq:trunc2}), (\ref{eq:103}), (\ref{eq:104}), (\ref{eq:105}), (\ref{eq:115}) imply (\ref{eq:toprove}), where
\begin{equation}
d'=d+1=O(e^{H_b(\min\{l/L,1/2\})L})
\end{equation}
satisfies Eq.~(\ref{eq:d}).
\end{proof}

\begin{lemma} \label{l:short}
If the spectrum of $H$ has non-degenerate gaps, then
\begin{equation}
\lim_{\tau\to\infty}\frac1\tau\int_0^\tau\|\phi(t)_A-\bar\phi_A\|_2^2\,\mathrm dt\le\frac{2^L}{D^\textnormal{eff}_h}.
\end{equation}
\end{lemma}

\begin{proof}
Using the same notation as in the proof of Lemma \ref{l:eq}, we choose $l=L$ so that $P_A^{\leqslant l}=\mathbbm 1_A$ and that $d=2^L$. Lemma \ref{l:short} follows from (\ref{eq:113}) and
\begin{equation}
\|\phi(t)_A-\bar\phi_A\|_2^2=\sum_{j,j'=1}^{d^2}c_j^*(t)c_{j'}(t)\tr(B_j^\dag B_{j'})=\sum_{j=1}^{d^2}|c_j(t)|^2.
\end{equation}
\end{proof}

\subsection{Proof of Theorem \ref{thm:half}}

\begin{proof} [Proof of (\ref{eq:10})]
Using Lemmas \ref{l:ndg} and \ref{l:eff}, it suffices to prove that
\begin{equation} \label{eq:2p}
\Pr_{t\in[0,\tau]}\left(\frac1m\sum_{j=1}^mS(\phi(t)_{A_j})=L\ln2-\frac{O(L\ln N)}N\right)=1-e^{-\Omega(N)}
\end{equation}
under the assumptions that (1) the spectrum of $H$ has non-degenerate gaps and (2)
\begin{equation}
S(\bar\phi)=N\ln2-O(\ln N),\quad D^\textnormal{eff}_h\ge2^N/\poly(N).
\end{equation}
Using the strong subadditivity of the von Neumann entropy,
\begin{equation} \label{eq:split}
\frac1m\sum_{j=1}^mS(\bar\phi_{A_j})\ge\frac{LS(\bar\phi)}N=L\ln2-\frac{O(L\ln N)}N.
\end{equation}
For $L\le N/10$, Lemma \ref{l:short} and (\ref{eq:12norm}) imply that
\begin{equation} \label{eq:pp}
\lim_{\tau\to\infty}\frac1\tau\int_0^\tau\frac1m\sum_{j=1}^m\|\phi(t)_{A_j}-\bar\phi_{A_j}\|_1\,\mathrm dt=e^{-\Omega(N)}.
\end{equation}
Markov's inequality implies that for sufficiently large $\tau$,
\begin{equation} \label{eq:prob}
\Pr_{t\in[0,\tau]}\left(\frac1m\sum_{j=1}^m\|\phi(t)_{A_j}-\bar\phi_{A_j}\|_1\le\frac1{\poly(N)}\right)=1-e^{-\Omega(N)}.
\end{equation}
Equation (\ref{eq:2p}) follows from (\ref{eq:split}), (\ref{eq:prob}), and the following lemma.
\end{proof}

\begin{lemma} [continuity of the von Neumann entropy \cite{Fan73, Aud07}] \label{l:cont}
Let $T:=\|\rho_1-\rho_2\|_1/2$ be the trace distance between two density operators $\rho_1,\rho_2$ on the Hilbert space $\mathbb C^D$. Then,
\begin{equation}
    |S(\rho_1)-S(\rho_2)|\le T\ln(D-1)-T\ln T-(1-T)\ln(1-T).
\end{equation}
\end{lemma}

Since by definition $0\le T\le1$, the right-hand side of this inequality is well defined.

\begin{proof} [Proof of (\ref{eq:50})]
Using Lemmas \ref{l:ndg} and \ref{l:eff}, it suffices to prove that
\begin{equation} \label{eq:124}
\Pr_{t\in[0,\tau]}\big(S_2(\phi(t)_A)=L\ln2-O(\ln N)\big)=1-e^{-\Omega(N/2-L)}/\poly(N)
\end{equation}
under the assumptions that (1) the spectrum of $H$ has non-degenerate gaps and (2)
\begin{equation} \label{eq:125}
S_2(\bar\phi)=\ln D^\textnormal{eff}_h=N\ln2-O(\ln N).
\end{equation}
Using the weak subadditivity \cite{LLZZ18} of the R\'enyi-$2$ entropy,
\begin{equation} \label{eq:126}
S_2(\bar\phi_A)\ge S_2(\bar\phi)-(N-L)\ln2=L\ln2-O(\ln N).
\end{equation}
Lemma \ref{l:short} and (\ref{eq:125}) imply that
\begin{equation}
\lim_{\tau\to\infty}\frac1\tau\int_0^\tau\|\phi(t)_A-\bar\phi_A\|_2^2\,\mathrm dt\le2^{L-N}\poly(N).
\end{equation}
Markov's inequality implies that for sufficiently large $\tau$,
\begin{equation} \label{eq:128}
\Pr_{t\in[0,\tau]}\big(\|\phi(t)_A-\bar\phi_A\|_2\le2^{-L/2}\poly(N)\big)=1-e^{-\Omega(N/2-L)}/\poly(N).
\end{equation}
Equation (\ref{eq:124}) follows from (\ref{eq:126}), (\ref{eq:128}), and
\begin{equation}
e^{-S_2(\phi(t)_A)/2}=\|\phi(t)_A\|_2\le\|\bar\phi_A\|_2+\|\phi(t)_A-\bar\phi_A\|_2=e^{-S_2(\bar\phi_A)/2}+\|\phi(t)_A-\bar\phi_A\|_2.
\end{equation}
\end{proof}

\subsection{Proof of Theorem \ref{thm:l}}

\begin{proof} [Proof of (\ref{eq:less})]
Using Lemmas \ref{l:eff} and \ref{l:eq}, it suffices to prove that
\begin{equation} \label{eq:suf}
\Pr_{t\in[0,\tau]}\left(\frac1m\sum_{j=1}^mS(\phi(t)_{A_j})\ge H_b(\nu)L-\frac{O(L\ln N)}N\right)=1-e^{-\frac{\Omega(N/2-L)^2}N}
\end{equation}
under the assumptions that
\begin{gather}
S(\bar\phi)\ge H_b(\nu)N-O(\ln N),\\
\lim_{\tau\to\infty}\frac1\tau\int_0^\tau\frac1m\sum_{j=1}^m\|\phi(t)_{A_j}-\bar\phi_{A_j}\|_1\,\mathrm dt\le e^{-(H_b(\nu)-c'\ln\frac{1-\nu}\nu)(\frac N2-L)}\poly(N)+e^{-\frac{\Omega(N/2-L)^2}L}.\label{eq:132}
\end{gather}
Using the strong subadditivity of the von Neumann entropy,
\begin{equation} \label{eq:enl}
\frac1m\sum_{j=1}^mS(\bar\phi_{A_j})\ge\frac{LS(\bar\phi)}N\ge H_b(\nu)L-\frac{O(L\ln N)}N.
\end{equation}
For $\nu=\Omega(1)$ and $L\le N/2-c\sqrt{N\ln N}$, (\ref{eq:132}) implies that
\begin{equation} 
\lim_{\tau\to\infty}\frac1\tau\int_0^\tau\frac1m\sum_{j=1}^m\|\phi(t)_{A_j}-\bar\phi_{A_j}\|_1\,\mathrm dt=e^{-\frac{\Omega(N/2-L)^2}N}.
\end{equation}
Markov's inequality implies that for sufficiently large $\tau$,
\begin{equation} \label{eq:last}
\Pr_{t\in[0,\tau]}\left(\frac1m\sum_{j=1}^m\|\phi(t)_{A_j}-\bar\phi_{A_j}\|_1\le\frac1{\poly(N)}\right)=1-e^{-\frac{\Omega(N/2-L)^2}N}.
\end{equation}
Equation (\ref{eq:suf}) follows from (\ref{eq:enl}), (\ref{eq:last}), and the continuity of the von Neumann entropy (Lemma \ref{l:cont}).
\end{proof}

\begin{proof} [Proof of (\ref{eq:half})]
We first consider the case that $L\le N/2-c\sqrt{N\ln N}$ for a sufficiently large constant $c$. Recall the definition (\ref{eq:PAdef}) of $P_A^{\leqslant l}$. Let $l=\nu L-O(\sqrt{L\ln N})$ with a sufficiently large constant hidden in the big-O notation. Similar to (\ref{eq:100}),
\begin{equation}
\Pr_h\big(\tr(\bar\phi_AP_A^{\leqslant l})\le1/\poly(N)\big)\ge1-1/\poly(N).
\end{equation}
Using this inequality and Lemmas \ref{l:eff} and \ref{l:eq}, it suffices to prove that
\begin{equation} \label{eq:137}
\Pr_{t\in[0,\tau]}\big(S(\phi(t)_A)\ge H_b(\nu)L-O(\sqrt{L\ln N})\big)=1-e^{-\frac{\Omega(N/2-L)^2}N}
\end{equation}
under the assumptions that
\begin{gather}
\tr(\bar\phi_AP_A^{\leqslant l})\le1/\poly(N),\quad S(\bar\phi)\ge H_b(\nu)N-O(\ln N),\label{eq:138}\\
\lim_{\tau\to\infty}\frac1\tau\int_0^\tau\|\phi(t)_A-\bar\phi_A\|_1\,\mathrm dt\le e^{-(H_b(\nu)-c'\ln\frac{1-\nu}\nu)(N/2-L)}\poly(N)+e^{-\frac{\Omega(N/2-L)^2}L}.\label{eq:139}
\end{gather}
$\bar\phi$ (\ref{eq:dep}) is a mixture of pure states, each of which has $n$ fermions. Hence,
\begin{gather}
\bar\phi^{>l}:=P_A^{>l}\bar\phi P_A^{>l}=P_{\bar A}^{<(n-l)}\bar\phi P_{\bar A}^{<(n-l)},\quad\bar\phi_A^{>l}:=\tr_{\bar A}\bar\phi^{>l}=P_A^{>l}\bar\phi_AP_A^{>l},\\
\tr(\bar\phi P_A^{\leqslant l})=\tr(\bar\phi_AP_A^{\leqslant l})=\|\bar\phi_A-\bar\phi_A^{>l}\|_1,\label{eq:141}
\end{gather}
where $P_{\bar A}^{<(n-l)}$ is defined in the same way as $P_A^{>l}$ (\ref{eq:PAdef}), and
\begin{align} \label{eq:147}
&\|\bar\phi-\bar\phi^{>l}\|_1\le\sum_{j=1}^{2^N}\big|\langle\phi|j\rangle_h\big|^2\cdot\big\||j\rangle_h\langle j|_h-P_A^{>l}|j\rangle_h\langle j|_hP_A^{>l}\big\|_1\nonumber\\
&=\sum_{j=1}^{2^N}\big|\langle\phi|j\rangle_h\big|^2\cdot\big\|P_A^{\leqslant l}|j\rangle_h\langle j|_h+P_A^{>l}|j\rangle_h\langle j|_hP_A^{\leqslant l}\big\|_1\le2\sum_{j=1}^{2^N}\big|\langle\phi|j\rangle_h\big|^2\sqrt{\langle j|_hP_A^{\leqslant l}|j\rangle_h}\nonumber\\
&\le2\sqrt{\sum_{j=1}^{2^N}\big|\langle\phi|j\rangle_h\big|^2\langle j|_hP_A^{\leqslant l}|j\rangle_h}=2\sqrt{\tr(\bar\phi P_A^{\leqslant l})}.
\end{align}
Using the weak subadditivity of the von Neumann entropy,
\begin{multline} \label{eq:142}
S(\bar\phi_A^{>l})\ge S(\bar\phi^{>l})-\ln\rank P_{\bar A}^{<(n-l)}=S(\bar\phi^{>l})-\ln\sum_{j=0}^{\nu(N-L)+O(\sqrt{L\ln N})}{N-L\choose j}\\
\ge S(\bar\phi^{>l})-H_b(\nu)(N-L)-O(\sqrt{L\ln N}).
\end{multline}
(\ref{eq:138}), (\ref{eq:141}), (\ref{eq:147}), (\ref{eq:142}), and the continuity of the von Neumann entropy imply that 
\begin{equation} \label{eq:143}
S(\bar\phi_A)\ge H_b(\nu)L-O(\sqrt{L\ln N}).
\end{equation}
For $\nu=\Omega(1)$ and $L\le N/2-c\sqrt{N\ln N}$, (\ref{eq:139}) implies that
\begin{equation} 
\lim_{\tau\to\infty}\frac1\tau\int_0^\tau\|\phi(t)_A-\bar\phi_A\|_1\,\mathrm dt=e^{-\frac{\Omega(N/2-L)^2}N}.
\end{equation}
Markov's inequality implies that for sufficiently large $\tau$,
\begin{equation} \label{eq:145}
\Pr_{t\in[0,\tau]}\big(\|\phi(t)_A-\bar\phi_A\|_1\le1/\poly(N)\big)=1-e^{-\frac{\Omega(N/2-L)^2}N}.
\end{equation}
Equation (\ref{eq:137}) follows from (\ref{eq:143}), (\ref{eq:145}), and the continuity of the von Neumann entropy.

If $N/2-c\sqrt{N\ln N}<L\le N/2$, let $A'$ be an (arbitrary) subsystem of $N/2-c\sqrt{N\ln N}$ fermionic modes in $A$. We have proved that
\begin{multline}
\Pr_h\left(\Pr_J\left(\Pr_{t\in[0,\tau]}\big(S(\phi(t)_{A'})\ge H_b(\nu)N/2-O(\sqrt{N\ln N})\big)\ge1-1/\poly(N)\right)=1\right)\\
\ge1-1/\poly(N).
\end{multline}
We complete the proof by noting that for any $t\in\mathbb R$,
\begin{equation}
S(\phi(t)_A)\ge S(\phi(t)_{A'})-\big(L-(N/2-c\sqrt{N\ln N})\big)\ln2\ge S(\phi(t)_{A'})-O(\sqrt{N\ln N}).
\end{equation}
\end{proof}

\section*{Acknowledgments} \addcontentsline{toc}{section}{Acknowledgments}

We would like to thank David A. Huse, Shao-Kai Jian, Xiao-Liang Qi for independently asking questions that motivated this work; Yingfei Gu, Pengfei Zhang, Tianci Zhou for explaining the path-integral approach in previous studies of entanglement dynamics in SYK models; Hong Liu for a discussion on the Page curve; and Jacobus J. M. Verbaarschot for pointing out the relationship between SYK models and embedded Gaussian ensembles \cite{MF75}. This material is based upon work supported by the U.S. Department of Energy, Office of Science, National Quantum Information Science Research Centers, Quantum Systems Accelerator.  AWH was also supported by NSF grants CCF-1729369 and PHY-1818914 and NTT (Grant AGMT DTD 9/24/20).

\paragraph{Data availability}There is no data associated with this manuscript.

\paragraph{Conflict of interest}The authors declare no competing interests.

\appendix

\section{Index of notation}

\small
\begin{longtable}{c|c|c}
notation & definition & defined in \\
\hline
$|0\rangle$ & vacuum state & after Eq.~(\ref{eq:ini}) \\
$\mathbbm 1_A$ & identity operator on subsystem $A$ & after Eq.~(\ref{eq:PAdef}) \\
$A$ & a subsystem smaller than or equal to half the system size & \S\ref{ss:et} \\
$\bar A$ & complement of subsystem $A$ (rest of the system) & \S\ref{ss:et} \\
$A_1,A_2,\ldots$ & a collection of subsystems & before (\ref{eq:phit}) \\
$B$ & a (not necessarily Hermitian) linear operator & before Eq.~(\ref{eq:bb}) \\
$\bar B$ & average of the expectation value of $B$ over time & Eq.~(\ref{eq:bb}) \\
$\Delta B$ & fluctuation of the expectation value of $B$ & Eq.~(\ref{eq:fl}) \\
$\mathbf C$ & correlation matrix between annihilation/creation operators & Eq.~(\ref{eq:def}) \\
$D^\textnormal{eff}_h$ & effective dimension of $|\phi\rangle$ & (\ref{eq:eff}) \\
$\{E_j\}$ & spectrum of a Hamiltonian that is clear from the context & Definition \ref{d:ndg} \\
$\e_{|A|=L}$ & average over all subsystems of $L$ fermionic modes & before Eq.~(\ref{eq:after-avg-def}) \\
$H$ & the model in which we prove entropy thermalization & Eq.~(\ref{eq:model}) \\
$H_\textnormal{SYK2}$ & Hamiltonian of the complex SYK2 model & Eq.~(\ref{eq:SYK2}) \\
$H_\textnormal{SYK4}$ & Hamiltonian of the complex SYK4 model & Eq.~(\ref{eq:SYK4}) \\
$H_\textnormal{SYK}$ & $H_\textnormal{SYK2}+\epsilon_2H_\textnormal{SYK4}$ & Eq.~(\ref{eq:SYK24}) \\
$H_\textnormal{d}$ & a model with random all-to-all density-density interactions & Eq.~(\ref{eq:dd}) \\
$H_b(\cdot)$ & binary entropy function & Table \ref{t:summary}, caption \\
$I_{N'}$ & identity matrix of size $N'\times N'$ & before Eq.~(\ref{eq:72}) \\
$J$ & a collection of random variables as coefficients of $H_\textnormal{SYK4}$ & Definition \ref{def:SYK4} \\
$K$ & a collection of random variables as coefficients of $H_\textnormal{d}$ & before Eq.~(\ref{eq:dd}) \\
$L$ & size of subsystem $A$ & \S\ref{s:rf}, 2nd paragraph \\
$\mathbf M$ & covariance matrix for Majorana operators & Eq.~(\ref{eq:covm}) \\
$\mathcal M_n$ & subspace of pure states with $n$ fermions & Eq.~(\ref{eq:Mn}) \\
$M_n$ & set of product states with $n$ spins up and $N-n$ spins down & \S\ref{app:ub}, 1st paragraph \\
$M'_n$ & orthonormal basis of $\Span M_n$ & \S\ref{app:ub}, 1st paragraph \\
$N$ & number of fermionic modes in the system & \S\ref{ss:m} \\
$N$ in \S\ref{sec:UB-proofs} & number of spins in the system & \S\ref{sec:UB-proofs}, 2nd paragraph \\
$[N]$ & set of integers from $1$ to $N$ & \S\ref{ss:m} \\
${[N]\choose n}$ & set of size-$n$ subsets of $[N]$ & Eq.~(\ref{eq:ncn}) \\
$O(\cdot)$ & Big-O notation & \S\ref{s:rf}, 1st paragraph\\
$\tilde O(\cdot)$ & Big-O with polylogarithmic corrections & \S\ref{s:rf}, 1st paragraph\\
$P_A^{>l},P_A^{\le l}$ & projector onto the subspace with $>l$ or $\le l$ fermions in $A$ & (\ref{eq:PAdef}) \\
$Q$ & fermion number operator & Eq.~(\ref{eq:Q}) \\
$Q_A$ & restriction of $Q$ to subsystem $A$ & Eq.~(\ref{eq:qa})\\
$R$ & a size-$n$ subset of $[N]$ & Eq.~(\ref{eq:ncn}) \\
$S(\cdot)$ & von Neumann entropy & Eq.~(\ref{eq:S}) \\
$S_\alpha(\cdot)$ & quantum R\'enyi entropy & Eq.~(\ref{eq:Sa}) \\
$U$ & a Haar-random unitary matrix & after Eq.~(\ref{eq:haar}) \\
$a$ & column vector $(a_1,a_2,\ldots,a_N)^T$ of annihilation operators & Definition \ref{def:SYK2} \\
$a^\dag$ & row vector $(a_1^\dag,a_2^\dag,\ldots,a_N^\dag)$ of creation operators & Definition \ref{def:SYK2} \\
$f$ & $L/N$ & Table \ref{t:summary}, caption \\
$h$ & random matrix from the Gaussian unitary ensemble & Definition \ref{def:SYK2} \\
$\{|j\rangle_h\}_{j=1}^{2^N}$ & eigenbasis of $H_\textnormal{SYK2}=a^\dag ha$ & near the end of \S\ref{ss:m} \\
$n$ & fermion number of the initial and time-evolved states & \S\ref{ss:ri}, 1st paragraph \\
$o(\cdot)$ & Little-O notation & \S\ref{s:rf}, 1st paragraph \\
$q$ & each term in the SYK$q$ Hamiltonian acts on $q$ sites & \S\ref{ss:m}, 2nd paragraph \\
$t$ & evolution time & (\ref{eq:24}) \\
$\Gamma(\cdot)$ & gamma function & \\
$\Theta(\cdot)$ & Big-Theta notation & \S\ref{s:rf}, 1st paragraph \\
$\Omega(\cdot)$ & Big-Omega notation & \S\ref{s:rf}, 1st paragraph \\
$\alpha$ & R\'enyi index & Definition \ref{def:ee} \\
$\beta$ & inverse temperature & Eq.~(\ref{eq:thermal}) \\
$\delta$ & Kronecker delta & \\
$\epsilon_1,\epsilon_2$ & infinitesimal parameters in the definition of $H$ & before Eq.~(\ref{eq:model}) \\
$\eta$ & column vector $(a_1,a_2,\ldots,a_N,a_1^\dag,a_2^\dag,\ldots,a_N^\dag)^T$ & (\ref{eq:etachi}) \\
$\mu_\alpha$ & $\alpha$-th moment of $|\langle\phi|j\rangle_h|^2$ & Eq.~(\ref{eq:ldef}) \\
$\nu$ & $n/N$, filling fraction & \S\ref{ss:ri}, 1st paragraph \\
$\sigma_\beta$ & thermal state at inverse temperature $\beta$ & Eq.~(\ref{eq:thermal}) \\
$\sigma_{\beta,A}$ & reduced density matrix of $\sigma_\beta$ for subsystem $A$ & Eq.~(\ref{eq:thr}) \\
$\tau$ & sufficiently large such that (\ref{eq:prob}), (\ref{eq:128}), (\ref{eq:last}), (\ref{eq:145}) hold & before Theorem \ref{thm:half} \\
$|\phi\rangle$ & initial state & Eq.~(\ref{eq:ini}) \\
$\phi(t)$ & density matrix of the time-evolved state $e^{-iHt}|\phi\rangle$ & (\ref{eq:phit}) \\
$\phi(t)_{A_j}$ & reduced density matrix of $\phi(t)$ for subsystem $A_j$ & (\ref{eq:phit}) \\
$\bar\phi$ & average of $\phi(t)$ over time & (\ref{eq:bphi}) \\
$\bar\phi_{A_j}$ & reduced density matrix of $\bar\phi$ for subsystem $A_j$ & (\ref{eq:bphi}) \\
$\chi$ & column vector $(\chi_1,\chi_2,\ldots,\chi_{2N})^T$ of Majorana operators & (\ref{eq:etachi}) \\
$|\psi\rangle_R$ & product state with modes in $R$ occupied & Eq.~(\ref{eq:psi}) \\
$\psi_R(t)$ & density matrix of the time-evolved state $e^{-iHt}|\psi\rangle_R$ & (\ref{eq:psit}) \\
$\psi_R(t)_A$ & reduced density matrix of $\psi_R(t)$ for subsystem $A$ & (\ref{eq:psit})
\end{longtable}

\printbibliography[heading=bibintoc]

\end{document}